%% file: main.tex
\documentclass[a4paper,fleqn]{llncs}

\input{content/preamble}

\usepackage{fullpage}
\usepackage{hyperref}

\title{Robust, Expressive, and Quantitative Linear Temporal Logics: Pick any Two for Free (full version)\thanks{Supported by the Saarbrücken Graduate School of Computer Science.}}

\author{Daniel Neider\inst{1}, Alexander Weinert\inst{2} and Martin Zimmermann\inst{3}}

\institute{Max Planck Institute for Software Systems, 67663 Kaiserslautern, Germany \\ \email{neider@mpi-sws.org} \and
German Aerospace Center (DLR), Institute for Software Technology, 51147 Cologne, Germany\\
\email{alexander.weinert@dlr.de}
\\
\and
University of Liverpool, Liverpool L69 3BX, United Kingdom\\
\email{martin.zimmermann@liverpool.ac.uk}
}

\begin{document}

\maketitle

\begin{abstract}
\input{content/abstract}
\end{abstract}

\section{Introduction}
\label{sec-intro}
\input{content/intro}

\section{Preliminaries}
\label{sec-prel}
\input{content/prel}

\subsection{Robust Linear Temporal Logic}
\label{subsec-briefrltl}
\input{content/prel-rltl}

\subsection{Linear Dynamic Logic}
\label{subsec-briefldl}
\input{content/prel-ldl}

\subsection{Prompt Linear Temporal Logic}
\label{subsec-briefprompt}
\input{content/prel-prompt}

\subsection{Prompt Linear Dynamic Logic}
\label{subsec-promptldl}
\input{content/prel-promptldl}

\section{Robust and Prompt Linear Temporal Logic}
\label{sec-rprompt}
\input{content/rprompt}

\subsection{Model Checking}
\label{subsec-rpromptresults-mc}
\input{content/rprompt-mc}

\subsection{Synthesis}
\label{subsec-rpromptresults-synt}

\input{content/rprompt-synt}

\section{Robust Linear Dynamic Logic}
\label{sec-rldl}
\input{content/rldl}

\subsection{Expressiveness}
\label{subsec-rldl-expressiveness}
\input{content/rldl-expressiveness}

\subsection{Model Checking and Synthesis}
\label{subsec-rldl-modelchecking}
\input{content/rldl-mcsynt}

\section{Towards Robust and Prompt Linear Dynamic Logic}
\label{sec-towardsrpldl}
\input{content/towardsrpromptldl}

\section{Conclusion}
\label{sec-conc}
\input{content/conclusion}

\bibliographystyle{splncs03}
\bibliography{biblio}

\end{document}

%% file: content/preamble.tex
\usepackage[utf8]{inputenc}
\usepackage[T1]{fontenc}
\usepackage{mathtools}
\usepackage{tikz}
\usetikzlibrary{automata}
\usepackage{amsmath}
\usepackage{amsfonts}
\usepackage{booktabs}
\usepackage{multirow}
\usepackage{wrapfig}
\usepackage{underscore}           
\usetikzlibrary{calc}
\usepackage{fixltx2e}

\newcommand{\myquot}[1]{``#1''}
\newcommand{\tabcenter}[1]{\multicolumn{1}{c}{#1}}

\renewcommand{\epsilon}{\varepsilon}
\newcommand{\bigo}[0]{\mathcal{O}}
\newcommand{\pow}[1]{2^{#1}}
\newcommand{\cceq}{::=}
\newcommand{\card}[1]{|#1|}
\newcommand{\size}[1]{|#1|}
\newcommand{\set}[1]{\{ #1 \}}
\newcommand{\nats}[0]{\mathbb{N}}

\newcommand{\bool}[0]{\mathbb{B}}
\renewcommand{\implies}{\rightarrow}
\newcommand{\cl}{\mathrm{cl}}
\newcommand{\ttrue}{\mathtt{tt}}
\newcommand{\ffalse}{\mathtt{ff}}

\newcommand{\init}[0]{I}

\newcommand{\marking}{t}
\newcommand{\aut}{\mathfrak{A}}
\newcommand{\autb}{\mathfrak{B}}
\newcommand{\autd}{\mathfrak{D}}
\newcommand{\autr}{\mathfrak{G}}
\newcommand{\autp}{\mathfrak{P}}
\newcommand{\pref}[2]{#1[0,#2)}
\newcommand{\suff}[2]{#1[#2,\infty)}
\newcommand{\accpar}{F}
\newcommand{\bplus}{\mathcal{B}^{+}}
\newcommand{\suc}[2]{\mathrm{Succ}_{#1}{#2}}

\newcommand{\ltl}{\text{LTL}}

\newcommand{\prompt}{\text{Prompt-LTL}}
\newcommand{\ldl}{\text{LDL}}
\newcommand{\rltl}{\text{rLTL}(\ensuremath{\Boxdot, \Diamonddot})}

\newcommand{\rprompt}{\text{rPrompt-LTL}}
\newcommand{\rldl}{\text{rLDL}}
\newcommand{\pldl}{\text{PLDL}}
\newcommand{\pltl}{\text{PLTL}}

\newcommand{\rpromptldl}{\text{rPrompt-LDL}}
\newcommand{\promptldl}{\text{Prompt-LDL}}

\newcommand{\prompteval}[0]{V^{\textsc{p}}}
\newcommand{\ldleval}[0]{V^{\textsc{d}}}
\newcommand{\rltleval}[0]{V^{\textsc{r}}}

\newcommand{\rprompteval}[0]{V^{\textsc{rp}}}
\newcommand{\rldleval}[0]{V^{\textsc{rd}}}

\newcommand{\rpromptldleval}[0]{V^{\textsc{rpd}}}

\newcommand{\promptldleval}[0]{V^{\textsc{pd}}}

\newcommand{\Rexp}{\mathcal{R}}
\newcommand{\robRexp}{\mathcal{R}^{\textsc{rd}}}
\newcommand{\robpromptRexp}{\mathcal{R}^{\textsc{rpd}}}

\newcommand{\conc}{\,;}

\newcommand{\tval}{\beta}

\newcommand{\halfthinspace}{{\kern .08333em}}

\newcommand{\ddiamond}[1]{\langle\/ #1 \/\rangle\,}

\newcommand{\bbox}[1]{[\halfthinspace#1\halfthinspace]\,}

\newcommand{\ddiamonddot}[1]{\langle\/ \cdot#1\cdot \/\rangle\,}

\newcommand{\promptddiamonddot}[1]{\langle\/ \cdot#1\cdot \/\rangle\textsubscript{\normalfont\textbf{p}}\,}

\newcommand{\promptddiamond}[1]{\langle\/ #1 \/\rangle\textsubscript{\normalfont\textbf{p}}\,}

\newcommand{\bboxdot}[1]{[\cdot#1\cdot]\,}

\tikzset{robust/.style={line width=.16ex,line join=round}}

\let\Box\relax
\DeclareMathOperator{\Box}{%
	\text{%
		\tikz[baseline]{%
    			\draw[robust] (0ex,-.1ex) -- (0ex, 1.4ex) -- (1.5ex, 1.4ex) -- (1.5ex, -.1ex) -- cycle;%
		}%
	}%
}

\DeclareMathOperator{\Boxdot}{%
	\text{%
		\tikz[baseline]{%
    			\draw[robust] (0ex, -.1ex) -- (0ex, 1.4ex) -- (1.5ex, 1.4ex) -- (1.5ex, -.1ex) -- cycle;%
	    		\fill (.75ex, .65ex) circle (.15ex);%
    		}%
	}%
}

\let\Diamond\relax
\DeclareMathOperator{\Diamond}{%
	\text{%
		\tikz[baseline]{%
			\draw[robust] (0ex,.6ex) -- (.95ex, 1.55ex) -- (1.9ex, .6ex) -- (.95ex, -.35ex) -- cycle;%
		}%
	}%
}

\DeclareMathOperator{\Diamonddot}{%
	\text{%
		\tikz[baseline]{%
			\draw[robust] (0ex,.6ex) -- (.95ex, 1.55ex) -- (1.9ex, .6ex) -- (.95ex, -.35ex) -- cycle;%
			\fill (.95ex, .6ex) circle (.15ex);%
		}%
	}%
}

\DeclareMathOperator{\R}{%
	\text{%
		\tikz[baseline]{%
			\node[inner sep=0pt, anchor=base, font=\bfseries] {R};%
		}%
	}%
}

\DeclareMathOperator{\Diamondprompt}{%
	\text{%
		\tikz[baseline]{%
			\draw[robust] (0ex,.6ex) -- (.95ex, 1.55ex) -- (1.9ex, .6ex) -- (.95ex, -.35ex) -- cycle;%
		}%
	}%
	\textsubscript{\normalfont\textbf{p}}%
}

\DeclareMathOperator{\Diamondpromptdot}{%
	\text{%
		\tikz[baseline]{%
			\draw[robust] (0ex,.6ex) -- (.95ex, 1.55ex) -- (1.9ex, .6ex) -- (.95ex, -.35ex) -- cycle;%
			\fill (.95ex, .6ex) circle (.15ex);%
		}%
	}%
	\textsubscript{\normalfont\textbf{p}}%
}

\DeclareMathOperator{\U}{%
	\text{%
		\tikz[baseline]{%
			\node[inner sep=0pt, anchor=base, font=\bfseries] {U};%
		}%
	}%
}

\DeclareMathOperator{\X}{%
	\text{%
		\tikz[baseline]{%
    			\draw[robust] (.75ex, .65ex) circle (.75ex);%
    		}%
	}%
}

\newcommand{\itotruthvalue}[1]{0^{#1-1}1^{5-#1}} 

\newcommand{\np}{\textsc{{NP}}}

\newcommand{\pspace}{\textsc{{PSpace}}}
\newcommand{\exptime}{\textsc{{ExpTime}}}
\newcommand{\twoexp}{\textsc{{2ExpTime}}}

\newcommand{\sys}{\mathcal{S}}
\newcommand{\trace}{\lambda}
\newcommand{\game}{\mathcal G}
\newcommand{\ggraph}{G}

%% file: content/abstract.tex
Linear Temporal Logic (LTL) is the standard specification language for reactive systems and is successfully applied in industrial settings.
However, many shortcomings of LTL have been identified in the literature, among them the limited expressiveness, the lack of quantitative features, and the inability to express robustness.
There is work on overcoming these shortcomings, but each of these is typically addressed in isolation.
This is insufficient for applications where all shortcomings manifest themselves simultaneously.

Here, we tackle this issue by introducing logics that address more than one shortcoming.
To this end, we combine the logics Linear Dynamic Logic, Prompt-LTL, and robust LTL, each addressing one aspect, to new logics.
For all combinations of two aspects, the resulting logic has the same desirable algorithmic properties as plain LTL.
In particular, the highly efficient algorithmic backends that have been developed for LTL are also applicable to these new logics. 
Finally, we discuss how to address all three aspects simultaneously.

%% file: content/intro.tex
Linear Temporal Logic ($\ltl$)~\cite{Pnueli77} is amongst the most prominent and most important specification languages for reactive systems, e.g., non-terminating controllers interacting with an antagonistic environment. Verification of such systems against $\ltl$ specifications is routinely applied in industrial settings nowadays~\cite{EisnerFismanPSL,Fix08}. Underlying this success story is the exponential compilation property~\cite{VardiWolper94}: every $\ltl$ formula can be effectively translated into an equivalent Büchi automaton of exponential size (and it turns out that this upper bound is tight). In fact, almost all verification algorithms for $\ltl$ are based on this property, which is in particular true for the popular polynomial space model checking algorithm and the doubly-exponential time synthesis algorithms. Other desirable properties of $\ltl$ include its compact and variable-free syntax and its intuitive semantics.

Despite the success of $\ltl$, a plethora of extensions of $\ltl$ have been studied, all addressing individual and specific shortcomings of $\ltl$, e.g., its limited expressiveness, its lack of quantitative features, and its inability to express robustness. Commonly, extensions of $\ltl$ as described above are only studied in isolation---the logics are either more expressive, or quantitative, or robust.
One notable exception is Parametric $\ldl$ ($\pldl$)~\cite{FaymonvilleZimmermann17}, which adds quantitative operators and increased expressiveness while maintaining the exponential compilation property and intuitive syntax and semantics.
In practical settings, however, it does not suffice to address one shortcoming of $\ltl$ while ignoring the others.
Instead, one needs a logic that combines multiple extensions while still maintaining the desirable properties of $\ltl$.
The overall goal of this paper is, hence, to bridge this gap, thereby enabling expressive, quantitative, and robust verification and synthesis. 

It is a well-known fact that $\ltl$ is strictly weaker than Büchi automata, i.e., it does not harness the full expressive power of the algorithmic backends.
Thus, increasing the expressiveness of $\ltl$ has generated much attention~\cite{LeuckerSanchez07,Vardi11,VardiWolper94,Wolper83} as it can be easily exploited:
as long as the new logic also has the exponential compilation property, the same optimized backends as for $\ltl$ can be used.
A prominent and recent example of such an extension that yields the full expressive power of Büchi automata is Linear Dynamic Logic ($\ldl$)~\cite{Vardi11}, which adds to $\ltl$ temporal operators guarded by regular expressions. 
In fact, the guarded operators can express all temporal operators of $\ltl$, i.e., we discard the temporal operators and only allow guarded operators.

As an example, consider the specification \myquot{p holds at every even time point, but may or may not hold at odd time points}. It is well-known that this property is not expressible in $\ltl$, as $\ltl$, intuitively, is unable to count modulo a fixed number. However, the specification is easily expressible in $\ldl$ as $\bbox{r} p$, where $r$ is the regular expression~$(\ttrue \cdot \ttrue)^*$. The formula requires $p$ to be satisfied at every position~$j$ such that the prefix up to position~$j$ matches the regular expression~$r$ (which is equivalent to $j$ being even), i.e., $\ttrue$ is an atomic regular expression that matches every letter. In this work, we consider $\ldl$ instead of the alternatives cited above for its conceptual simplicity: $\ldl$ has a simple and variable-free syntax based on regular expressions as well as intuitive semantics (assuming some familiarity with regular expressions).

Another serious shortcoming of $\ltl$ (and $\ldl$) is its inability to adequately express timing bounds. For example, consider the specification \myquot{every request~$q$ is eventually answered by a response~$p$}, which is expressed in $\ltl$ as $\Box (q \rightarrow \Diamond p)$.
It is satisfied, even if the waiting time between requests~$q$ and responses~$p$ diverges to infinity, although such a behavior is typically undesired. Again, a long line of research has addressed this second shortcoming of $\ltl$~\cite{AlurEtessamiLaTorrePeled01,FaymonvilleZimmermann17,Koymans90,KupfermanPitermanVardi09,Zimmermann18}.
The most basic one is $\prompt$~\cite{KupfermanPitermanVardi09}, which adds the prompt-eventually operator~$\Diamondprompt$ to $\ltl$.
To retain decidability~\cite{AlurEtessamiLaTorrePeled01}, one has to give up negation and implication when adding the prompt-eventually operator. 
This is no restriction for $\ltl$, as every formula has an equivalent one in negation normal form. 

The semantics is now defined with an additional parameter~$k$, which bounds the scope of $\Diamondprompt$: $\Box (q \rightarrow \Diamondprompt p)$ requires every request~$q$ to be answered within $k$ steps, when evaluated with respect to $k$.
The resulting logic is a quantitative one: either one quantifies the parameter~$k$ existentially and obtains a boundedness problem, e.g., \myquot{does there exist a bound~$k$ such that every request can be answered within $k$ steps}, or one even aims to determine the optimal bound~$k$.
Again, $\prompt$ retains the desirable properties of $\ltl$, i.e., the exponential compilation property as well as intuitive syntax and semantics. Furthermore, $\prompt$ captures the technical core of the alternatives cited above, e.g., decision problems for the more general logic~$\pltl$~\cite{AlurEtessamiLaTorrePeled01} can be reduced to those for $\prompt$. For these reasons, we study $\prompt$ in this work.

Finally, a third line of extensions of $\ltl$ is concerned with the concept of robustness, which is much harder to formalize. This is reflected by a multitude of incomparable notions of robustness in verification~\cite{DBLP:conf/cav/BloemCGHJ10,DallalNeiderTabuada16,DBLP:conf/formats/DonzeM10,DBLP:conf/acsd/DoyenHLN10,DBLP:journals/tcs/FainekosP09,DBLP:conf/rtss/MajumdarS09,NeiderWeinertZimmermann18,DBLP:journals/tac/TabuadaCRM14,TabuadaNeider16}.
Here, we are interested in robust $\ltl$ ($\text{r}\ltl$)~\cite{TabuadaNeider16}, which equips $\ltl$ with a five-valued semantics that captures different degrees of violations of universal specifications.
As an example, consider the specification \myquot{if property $\varphi$ always holds true, then property $\psi$ also always holds true}, which is expressed in $\ltl$ as $\Box \varphi \rightarrow \Box \psi$ and is typical for systems that have to interact with an antagonistic environment.
In classical semantics, the whole formula is satisfied as soon as the assumption~$\varphi$ is violated once, even if the guarantee~$\psi$ is violated as well.
By contrast, the semantics of robust $\ltl$ ensures that the degree of the violation of $\Box \psi$ is always proportional to the degree of the violation of $\Box \varphi$.
To this end, the degree of a violation of a property $\Box \varphi$ is expressed by five different truth values: either $\varphi$ always holds, or $\varphi$ is violated only finitely often, violated infinitely often, violated almost always, or violated always.
Again, robust $\ltl$ has the exponential compilation property and an intuitive syntax (though its semantics is more intricate). In this work, we consider robust $\ltl$, as it is the first logic that intrinsically captures the notion of robustness in $\ltl$. In particular, formulas of robust $\ltl$ are evaluated over traces with Boolean truth values for atomic propositions and do not require non-Boolean assignments, which are often hard to determine in real-life applications.

We consider here the fragment~$\rltl$ of $\text{r}\ltl$ that only contains the temporal operators eventually and always, as it already captures the most interesting aspects of robustness. 
	
\subsection{Our Contributions}
\label{sec:contributions}

In this paper, which is an extended version of an earlier conference paper~\cite{DBLP:journals/corr/abs-1909-08538}, we develop logics that address more than one shortcoming of $\ltl$ at a time. See Figure~\ref{fig:logics} for an overview. In comparison to the conference version, we have added Section~\ref{sec-prel} introducing the logics we combine, all proofs omitted due to space restrictions, and Section~\ref{subsec-fragment}.

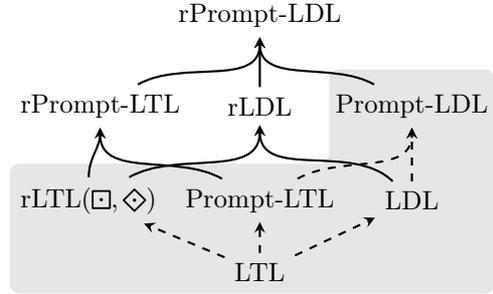
\begin{figure} \centering
	\begin{tikzpicture}[thick]
		
		\draw[rounded corners,fill=black!10,draw=white]
			(-5.3, .3) |- (2,2.5) |- (5.3,4.5) |- (0,.3) -- cycle;
	
		\node[align=center] (ltl) at (0,0.8) {\ltl};
		
		\begin{scope}[shift={(0,2)}]
			\node[align=center] (rltl) at (-4,0) {\rltl};
			\node[align=center] (promptltl) at (0,0) {\prompt};
			\node[align=center] (ldl) at (4,0) {\ldl};
		\end{scope}
		
		\begin{scope}[shift={(0,4)}]
			\node[align=center] (rpromptltl) at (-4,0) {\rprompt};
			\node[align=center] (rldl) at (0,0) {\rldl};
			\node[align=center] (promptldl) at (4,0) {\promptldl};
		\end{scope}
		
		\node[align=center] (rpromptldl) at (0,5.5) {\rpromptldl};
		
		\path[-stealth,]
			(ltl)
				edge[dashed] (rltl)
				edge[dashed] (promptltl)
				edge[dashed] (ldl)
			(rltl)
				edge[out=90,in=-90] (rpromptltl)
				edge[out=30,in=-90] (rldl)
			(promptltl)
				edge[out=150,in=-90] (rpromptltl)
				edge[dashed,out=30,in=-90] (promptldl)
			(ldl)
				edge[out=135,in=-90] (rldl)
				edge[dashed] (promptldl)
			(rpromptltl)
				edge[out=30,in=-90] (rpromptldl)
			(rldl)
				edge (rpromptldl)
			(promptldl)
				edge[out=150,in=-90] (rpromptldl);
	\end{tikzpicture}
	\caption{The logics studied in this work. Existing logics and influences are marked gray with dashed arrows.}
	\label{fig:logics}
\end{figure}

	In Section~\ref{sec-rprompt}, we ``robustify'' $\prompt$. More precisely, we introduce a novel logic, named $\rprompt$, by extending the five-valued semantics from robust $\ltl$ to $\prompt$. Our main result here shows that $\rprompt$ retains the exponential compilation property. 
Then, in Section~\ref{sec-rldl}, we ``robustify'' $\ldl$: we introduce a novel logic, named $\rldl$, by lifting the five-valued semantics of robust $\ltl$ to $\ldl$. Our main result shows that $\rldl$ also retains the exponential compilation property. 
Hence, one can indeed combine any two of the three extensions of $\ltl$ while still preserving the desirable algorithmic properties of $\ltl$. In particular, let us stress again that all highly sophisticated algorithmic backends developed for $\ltl$ are applicable to these novel logics as well, e.g., we show that the verification problem and the synthesis problem for each of these logics is solvable without an (asymptotic) increase in complexity.

Tabuada and Neider gave two proofs showing that robust $\ltl$ has the exponential compilation property. The first one presented a translation of robust $\ltl$ into equivalent Büchi automata of exponential size while the second one is based on a polynomial translation of robust $\ltl$ into (standard) $\ltl$, which is known to be translatable into equivalent Büchi automata of exponential size. We refer to those two approaches as the \emph{direct} approach and the \emph{reduction-based} approach. To obtain our results mentioned above, we need to generalize both. 
To prove the exponential compilation property for $\rldl$, we generalize the direct approach by exhibiting a direct translation of $\rldl$ into Büchi automata via alternating automata. In contrast, to prove the exponential compilation property for $\rprompt$, we present a generalization of the reduction-based approach translating $\rprompt$ into equivalent $\prompt$ formulas of linear size, which have the exponential compilation property.  
	
	 Finally, in Section~\ref{sec-towardsrpldl}, we discuss the combination of all three aspects. Recall that we present a direct translation to automata for $\rldl$ and a reduction-based one for $\rprompt$. For reasons we discuss in Section~\ref{sec-towardsrpldl}, it is challenging to develop a reduction from $\rldl$ to $\ldl$ or a direct translation for $\rprompt$ that witness the exponential compilation property. Hence, both approaches seem inadequate to deal with the combination of all three extensions. Ultimately, we leave the question of whether the logic combining all three aspects has the exponential compilation property for future work.

%% file: content/prel.tex
We denote the non-negative integers by~$\nats$, the set~$\set{0,1}$ of Boolean truth values by~$\bool$, and the power set of~$S$ by~$\pow{S}$.
By convention, we have $\min \emptyset = 1$ and $\max \emptyset = 0$ when the operators range over subsets of $\bool$. Following Tabuada and Neider~\cite{TabuadaNeider16}, the set of truth values for robust semantics is $\bool_4 = \set{0000, 0001, 0011, 0111, 1111}$, which are ordered by $0000 \prec 0001 \prec 0011 \prec 0111 \prec 1111$.
We write $\preceq$ for the non-strict variant of $\prec$ and define $\min \emptyset = 1111$ and $\max \emptyset = 0000$ when the operators range over subsets of $\bool_4$.

Throughout this work, we fix a finite non-empty set~$P$ of atomic propositions and define the shorthands $\ttrue = p \vee \neg p$ and $\ffalse = p \wedge \neg p$ for some atomic proposition~$p$.
For a set $A \subseteq P$ and a propositional formula~$\phi$ over $P$, we write $A \models \phi$ if the variable valuation mapping elements in $A$ to $1$ and elements in $P\setminus A$ to $0$ satisfies $\phi$.
A trace (over $P$) is an infinite sequence~$w \in (\pow{P})^\omega$. Given a trace~$w = w(0) w(1) w(2) \cdots$ and a position~$j \in \nats$, we define $\pref{w}{j} = w(0) \cdots w(j-1)$ and $\suff{w}{j} = w(j) w(j+1) w(j+2) \cdots$, i.e., $\pref{w}{j}$ is the prefix of length~$j$ of $w$ and $\suff{w}{j}$ the remaining suffix. In particular, $\pref{w}{0}$ is empty  and $\suff{w}{0}$ is $w$.

In the remainder of this section, we introduce the logics we generalize in this work, namely Robust Linear Temporal Logic ($\rltl$)~\cite{TabuadaNeider16}, Linear Dynamic Logic ($\ldl$)~\cite{Vardi11}, and Prompt Linear Temporal Logic ($\prompt$)~\cite{KupfermanPitermanVardi09}.
Also, we introduce Prompt Linear Dynamic Logic ($\promptldl$)~\cite{FaymonvilleZimmermann17}, which we use in some proofs. 
See Table~\ref{tab:logics} for an overview.
References for the results mentioned in the table are given in the following subsections introducing the logics.

\begin{table}[h]\centering
	\begin{tabular}{llll} \toprule
		\tabcenter{\multirow{2}{*}{Logic}} & \tabcenter{\multirow{2}{*}{Operators}} & \multicolumn{2}{c}{Complexity} \\
		 &  & \tabcenter{Model Checking} & \tabcenter{Synthesis} \\ \midrule
		\rltl & $\neg, \land, \lor, \implies, \Boxdot, \Diamonddot$ & \np-hard/in \pspace &  \twoexp-compl. \\
		\ldl & $\neg, \land, \lor, \implies, \ddiamond{r}, \bbox{r}$ & \pspace-compl. & \twoexp-compl. \\
		\prompt & $\land, \lor, \X, \U, \R, \Diamondprompt$ & \pspace-compl. & \twoexp-compl. \\
		\promptldl & $\land, \lor, \ddiamond{r}, \bbox{r}, \promptddiamond{r}$ & \pspace-compl. & \twoexp-compl. \\ \bottomrule
	\end{tabular}
	\caption{The logics our work is based on.}
	\label{tab:logics}
\end{table}

 We define the semantics of all these logics by evaluation functions~$V$ mapping a trace, a formula, and a bound (in the case of a quantitative logic) to a truth value. This is prudent for robust semantics, as it allows us to introduce useful notation naturally. For the sake of consistency, we also use this approach for the other logics, whose semantics is typically defined via satisfaction relations. In particular, $\rltleval$, $\ldleval$, and $\prompteval$ denote the evaluation functions of $\rltl$, $\ldl$, and $\prompt$, respectively. Nevertheless, our definitions here are equivalent to the original definitions.

%% file: content/prel-rltl.tex
The main impetus behind the introduction of robust $\ltl$ was the need to capture the concept of robustness in temporal logics.
As a first motivating example consider the $\ltl$ formula~$\Box p$, stating that $p$ holds at every position.
Consequently, the formula is violated if there is a single position where $p$ does not hold.
However, this is a very \myquot{mild} violation of the property and there are much more \myquot{severe} violations.
As exhibited by Tabuada and Neider, there are four canonical \emph{degrees} of violation of $\Box p$:
(i) $p$ is violated at finitely many positions,
(ii) $p$ is violated at infinitely many positions,
(iii) $p$ is violated at all but finitely many positions, and
(iv) $p$ is violated at all positions.
These first three degrees are captured by the $\ltl$ formulas $\Diamond\Box p$, $\Box \Diamond p$, and $\Diamond p$, which are all weakenings of $\Box p$.
All five possibilities, satisfaction and four degrees of violation, are captured in robust $\ltl$ by the truth values 
\[ 1111 \succ 0111 \succ 0011 \succ 0001 \succ 0000 \]
introduced above.
By design, the formula~$\Boxdot p$ of robust $\ltl$\footnote{Following the precedent for robust $\ltl$, we use dots to distinguish operators of robust logics from those of classical logics throughout the paper.} has
\begin{itemize}
	\item truth value $1111$ on all traces where $p$ holds at all positions,
	\item truth value $0111$ on all traces where $p$ holds at all but finitely many positions,
	\item truth value $0011$ on all traces where $p$ holds at infinitely many positions and does not hold at infinitely many positions,
	\item truth value $0001$ on all traces where $p$ only holds at finitely many positions, and 
	\item truth value $0000$ on all traces where $p$ holds at no position.
\end{itemize}

As a further example, consider the formula~$\Boxdot p \rightarrow \Boxdot q$. For this formula, the robust semantics captures the intuition described in the introduction: the implication is satisfied (i.e., has truth value $1111$), if  the degree of violation of the property~\myquot{always q} is at most the degree of violation  of the property~\myquot{always p}. Thus, if $p$ is violated finitely often, then $q$ may also be violated finitely often (but not infinitely often) while still satisfying the implication.

Formally, the formulas of $\rltl$ are given by the grammar
\[\varphi \cceq p \mid \neg \varphi \mid \varphi \wedge \varphi \mid \varphi \vee \varphi \mid \varphi \implies \varphi \mid \Diamonddot \varphi 
  \mid \Boxdot \varphi,
\]
where $p$ ranges over the atomic propositions in $P$.
Note that the syntax of $\rltl$ explicitly contains implication and conjunction; due to the many-valued semantics of $\rltl$ introduced below, these two operators cannot be recovered from disjunction and negation.
We define the size~$\size{\varphi}$ of a formula as the number of distinct subformulas of $\varphi$. 

Intuitively, conjunction and disjunction are defined as usual using minimization and maximization relying on the order of truth values indicated above while negation is based on the intuition that $1111$ represents satisfaction and all other truth values represent degrees of violation.
Hence, a negation~$\neg \varphi$ is satisfied (i.e., has truth value $1111$), if $\varphi$ has truth value less than $1111$, and it is violated (i.e., has truth value $0000$) if $\varphi$ has truth value~$1111$.
Finally, the semantics of the eventually operator is defined as usual, i.e., the truth value of $\Diamonddot \varphi$ on $w$ is the maximal truth value that is assumed by $\varphi$ on some suffix of $w$.

This intuition is formalized in the evaluation function~$\rltleval$, which maps a trace~$w \in (\pow{P})^\omega$ and an $\rltl$ formula $\varphi$ to a truth value~$\rltleval(w,\varphi)$ in $\bool_4 $ and which is defined as follows~\cite{TabuadaNeider16}:

\label{rltlsemdef}
\begin{itemize}
\item $\rltleval(w,p) =\begin{cases}
	1111 &\text{if $p \in w(0)$,}\\
	0000 &\text{if $p \notin w(0)$,}
\end{cases}$ \quad\quad $\text{\normalfont\bfseries --} \,\,\,\rltleval(w,\neg \varphi) = \begin{cases} 
 	1111 &\text{if $\rltleval(w,\varphi) \neq 1111$,}\\ 
 	0000 &\text{if $\rltleval(w,\varphi) = 1111$,} 
 \end{cases}$	

\item
$\rltleval(w, \varphi_0 \wedge \varphi_1) = \min\set{\rltleval(w, \varphi_0), \rltleval(w,\varphi_1) }$, 
\item $\rltleval(w, \varphi_0 \vee \varphi_1) = \max\set{\rltleval(w, \varphi_0), \rltleval(w,\varphi_1) }$, 

\item
$\rltleval(w, \varphi_0 \implies \varphi_1) = \begin{cases}
 	1111 &\text{if $\rltleval(w,\varphi_0) \preceq \rltleval(w, \varphi_1)$,}\\
 	\rltleval(w, \varphi_1) &\text{if $\rltleval(w,\varphi_0) \succ \rltleval(w, \varphi_1)$,}
 \end{cases}$
 \item $\rltleval(w, \Diamonddot \varphi ) = b_1 b_2 b_3 b_4$ with $b_i = \max_{j \ge 0} \rltleval_i(\suff{w}{j}, \varphi )$, and 
 \item $\rltleval(w, \Boxdot \varphi) = b_1 b_2 b_3 b_4 $ with 
 	\begin{itemize}
 		\item $b_1 = \min_{j \ge 0} \rltleval_1(\suff{w}{j}, \varphi )$,
 		\item $b_2 = \max_{j' \in \nats} \min_{j \ge j'}\rltleval_2(\suff{w}{j}, \varphi )$,
 		\item $b_3 = \min_{j' \in \nats} \max_{j \ge j'}\rltleval_3(\suff{w}{j}, \varphi )$, and
 		\item $b_4 = \max_{j \ge 0} \rltleval_4(\suff{w}{j}, \varphi )$.
 	\end{itemize}
\end{itemize}
Here, $\rltleval_i(w, \varphi)$ denotes the projection of $\rltleval(w, \varphi)$ to its $i$-th component, i.e., we have 
\[\rltleval(w, \varphi) = \rltleval_1(w, \varphi)\rltleval_2(w, \varphi)\rltleval_3(w, \varphi)\rltleval_4(w, \varphi).\]

The first bit of the semantics captures the classical semantics of $\ltl$, i.e., we have $\rltleval_1(w,\varphi) = 1$ if and only if $w$ satisfies $\varphi$ classically.
Intuitively, the next three bits are obtained by weakening the semantics of the subformulas of the form~$\Boxdot \varphi$:
instead of (classically) requiring every position to satisfy $\varphi$, the second bit is one if almost all positions satisfy $\varphi$ (i.e., $\Diamond \Box \varphi$ holds), the third bit is one if infinitely many positions satisfy $\varphi$ (i.e., $\Box \Diamond \varphi$ holds), and the fourth bit is one if at least one position satisfies $\varphi$ (i.e., $\Diamond \varphi$ holds).
Note that the semantics of negation and implication are also non-classical and break the intuition given above, e.g., for formulas of the form~$\neg \Boxdot\varphi$.
For a full motivation and explanation of the semantics, we refer to the original work introducing $\rltl$~\cite{TabuadaNeider16} as well as follow-up work~\cite{AnevlavisPNT18,AnevlavisNPT19,monitoring}.

Verification problems with $\rltl$ specifications have a threshold $\tval \in \bool_4$ as an additional input and ask every trace to evaluate to at least~$\tval$.
Tabuada and Neider showed that the model checking problem with $\rltl$ specifications can be solved in polynomial space\footnote{Tabuada and Neider only showed that their algorithm runs in exponential time, but using standard on-the-fly techniques~\cite{VardiWolper94} it can also be implemented in polynomial space.} and that infinite games with $\rltl$ specifications can be solved in doubly-exponential time~\cite{TabuadaNeider16}.
The lower bounds presented in Table~\ref{tab:logics} are derived from the special case of $\ltl(\Box,\Diamond)$~\cite{gameswithboxes,AlurLatorre04}, which is a fragment of $\rltl$.

When Tabuada and Neider introduced robust $\ltl$, they first considered the fragment~$\rltl$ without the next, until, and release operators, which already captures the most interesting problems arising from adding robustness.
Then, they added the missing operators and studied the full logic~\cite{TabuadaNeider16}.
Here, we follow their approach and only consider generalizations of the fragment~$\rltl$ that only contains the temporal operators~$\Boxdot$ and~$\Diamonddot$.
We comment on the effect of this restriction when defining the combinations of logics. 

%% file: content/prel-ldl.tex
The logic $\ldl$ has only two temporal operators, $\ddiamond{r}$ and $\bbox{r}$, which can be understood as guarded variants of the classical eventually and always operators from $\ltl$, respectively.
Both are guarded by regular expressions~$r$ over the atomic propositions that may contain tests, which are again $\ldl$ formulas.
These two operators together with Boolean connectives capture the full expressive power of the $\omega$-regular expressions, i.e., $\ldl$ exceeds the expressiveness of $\ltl$. 

The formulas of $\ldl$ are given by the grammar
\begin{align*}
\varphi &\cceq p \mid \neg \varphi \mid \varphi \wedge \varphi \mid \varphi \vee \varphi \mid \varphi \implies \varphi \mid \ddiamond{r} \varphi 
  \mid \bbox{r} \varphi \\
    r & \cceq \phi \mid \varphi? \mid r+r \mid r \conc r \mid r^*
\end{align*}
where $p$ ranges over the atomic propositions in $P$ and where $\phi$ ranges over arbitrary propositional formulas over $P$.
The regular expressions have two types of atoms: propositional formulas~$\phi$ over the atomic propositions and tests~$\varphi?$, where $\varphi$ is again an $\ldl$ formula.
As we will see later, the semantics of these two kinds of atoms differ significantly.
We refer to formulas of the form~$\ddiamond{r}\varphi$ and $\bbox{r}\varphi$ as diamond formulas and box formulas, respectively.
In both cases, we call $r$ the guard of the operator. 

We denote the set of subformulas of $\varphi$ by $\cl(\varphi)$.
Guards are not subformulas, but the formulas appearing in the tests are, e.g., we have 
	\[ \cl(\ddiamond{\halfthinspace p?\conc q}p') = \set{\halfthinspace p, p', \ddiamond{\halfthinspace p?\conc q}p'}. \]
The size~$\card{\varphi}$ of $\varphi$ is the sum of $\card{\cl(\varphi)}$ and the sum of the lengths of the guards appearing in $\varphi$ (measured in the number of operators),  taking each occurrence of a guard in the syntax tree (not the syntax DAG like for subformulas) into account.

Formally, a formula~$\ddiamond{r}\varphi$ is satisfied by a trace~$w$, if there is some~$j$ such that the prefix~$\pref{w}{j}$ matches the regular expression~$r$ and the corresponding suffix~$\suff{w}{j}$ satisfies $\varphi$.
Dually, a formula~$\bbox{r}\varphi$ is satisfied by a trace~$w$ if for every~$j$ with $\pref{w}{j}$ matching $r$, $\suff{w}{j}$ satisfies $\varphi$.
Thus, while the classical eventually and always operator range over all positions, the operators of $\ldl$ range only over those positions whose induced prefix matches the guard of the operator.

Analogously to the definition for $\rltl$, and slightly non-standard, we define the semantics of $\ldl$ by specifying an evaluation function~$\ldleval$ mapping a trace~$w$ and a formula~$\varphi$ to a truth value from~$\bool$ denoting whether $w$ satisfies $\varphi$ or not.
Also, our presentation of the semantics here is slightly cumbersome, in particular the definition for the implication, again to align with the definition for $\rltl$.
Nevertheless, our definition below is equivalent to the classical semantics of $\ldl$ (cf.~\cite{Vardi11,GiacomoVardi13,FaymonvilleZimmermann17}) via a satisfaction relation~$\models$ in the following sense:
we have $\ldleval(w, \varphi) = 1$ if and only if $w \models \varphi$.
\begin{itemize}
	\item $\ldleval(w,p) = \begin{cases}
		1 &\text{if $p \in w(0)$,} \\
		0 &\text{if $p \notin w(0)$,}
	\end{cases}$ \quad\quad $\text{\normalfont\bfseries --} \,\,\,\ldleval(w,\neg\varphi) = \begin{cases}
		1 &\text{if $\ldleval(w, \varphi)=0$,}\\
		0 &\text{if $\ldleval(w, \varphi)=1$,}
	\end{cases}$
	\item $\ldleval(w,\varphi_0 \wedge \varphi_1) = \min\set{\ldleval(w, \varphi_0), \ldleval(w, \varphi_1)}$,
	\item $\ldleval(w,\varphi_0 \vee \varphi_1) = \max\set{\ldleval(w, \varphi_0), \ldleval(w, \varphi_1) }$,
	\item $\ldleval(w,\varphi_0 \implies \varphi_1) = \begin{cases}
		1 &\text{if $\ldleval(w, \varphi_0) \le \ldleval(w, \varphi_1)$,}\\
		\ldleval(w, \varphi_1) &\text{if $\ldleval(w, \varphi_0) > \ldleval(w, \varphi_1)$,}
	\end{cases}$
	\item $\ldleval(w,\ddiamond{r}\varphi) = \max_{j \in \Rexp(w,r)} \ldleval(\suff{w}{j}, \varphi)$, and
	\item $\ldleval(w,\bbox{r}\varphi) = \min_{j \in \Rexp(w,r)} \ldleval(\suff{w}{j}, \varphi)$.
\end{itemize}

Here, the match set~$\Rexp(w,r) \subseteq \nats$ contains all positions~$j$ such that $\pref{w}{j}$ matches $r$. Recall that $\pref{w}{j}$ denotes the prefix of $w$ of length~$j$, i.e., $\pref{w}{j} =w(0)\cdots w(j-1)$. In particular, $\pref{w}{0}$ is empty and $\pref{w}{\infty}$ is~$w$. Now, $\Rexp(w,r)$ is defined inductively as follows:
\begin{itemize}

\item $\Rexp(w,\phi) = \set{1}$ if $w(0) \models \phi$ (i.e., we evaluate $\phi$ in standard Boolean semantics) and $\Rexp(w,\phi) =  \emptyset$ otherwise, for propositional~$\phi$.

\item $\Rexp(w,\varphi?) = \set{0}$ if $\ldleval(w, \varphi)=1$ and $\Rexp(w,\varphi?) =  \emptyset$ otherwise.

\item $\Rexp(w,r_0 + r_1) = \Rexp(w,r_0) \cup \Rexp(w,r_1)$.

\item $\Rexp(w,r_0 \conc r_1) = \set{j_0 + j_1 \mid j_0, j_1 \ge 0 \text{ and } j_0 \in \Rexp(w,r_0) \text{ and } j_1 \in \Rexp(\suff{w}{j_0}, r_1)}$. Thus, for $j$ to be in $\Rexp(w,r_0 \conc r_1)$, it has to be the sum of natural numbers $j_0$ and $j_1$ such that $w$ has a prefix of length $j_0$ that matches $r_0$ and $\suff{w}{j_0}$ has a prefix of length~$j_1$ that matches $r_1$. 

\item $\Rexp(w,r^*) = \set{0} \cup \set{j_1 + \cdots + j_\ell \mid 0 \le j_{\ell'} \in \Rexp(\suff{w}{j_1 + \cdots + j_{\ell'-1}},r) \text{ for all } \ell'\in\set{1, \ldots, \ell}}$, where we use $j_1 + \cdots + j_{0} = 0$.
Thus, for $j$ to be in $\Rexp(w,r^*)$, it has to be expressible as $j = j_1 +\cdots + j_\ell$ with non-negative~$j_{\ell'}$ such that the prefix of $w$ of length~$j_1$ matches $r$, the prefix of length~$j_2$ of $\suff{w}{j_1}$ matches $r$, and in general, the prefix of length~$j_{\ell'}$ of $\suff{w}{j_1+\cdots +j_{\ell'-1}}$ matches $r$, for every $\ell' \in \set{1, \ldots, \ell}$.

\end{itemize}

Due to tests, membership of some $j$ in $\Rexp(w,r)$ does, in general, not only depend on the prefix~$\pref{w}{j}$, but on the complete trace~$w$.
Also, the semantics of the propositional atom~$\phi$ differ from the semantics of the test~$\phi?$: the former consumes an input letter, while tests do not.
Hence, the guards of $\ldl$ feature both kinds of atoms. 

Fix some trace~$w$, a formula~$\varphi$, and a guard~$r$.
We say that a position~$j$ of $w$ is an $r$-match if $j \in \Rexp(w,r)$.
Further, $j$ is a $\varphi$-satisfying position of $w$ if $\ldleval(\suff{w}{j}, \varphi)=1$.
Thus, the formula~$\ddiamond{r}\varphi$ requires some $\varphi$-satisfying $r$-match to exist.
Dually, $\bbox{r}\varphi$ requires every $r$-match of $w$ to be $\varphi$-satisfying (in particular, this is the case if there is no $r$-match in $w$).
Thus, the diamond operator generalizes the eventually operator and the box operator generalizes the always operator, which are the respective special cases for a trivial guard that matches every position, e.g., $\ttrue^*$.
Similarly, the next, until, and release operator of $\ltl$ can be expressed in $\ldl$ (the latter two use tests in the guards).
Thus, $\ltl$ is a fragment of $\ldl$.
Furthermore, it is known that $\ldl$ captures the $\omega$-regular languages~\cite{Vardi11}.

Model checking against $\ldl$ specifications is $\pspace$-complete and solving $\ldl$ games is $\twoexp$-complete~\cite{FaymonvilleZimmermann17,Vardi11}.

%% file: content/prel-prompt.tex
To express timing constraints, the logic $\prompt$ adds the prompt-eventually operator~$\Diamondprompt$ to $\ltl$.
Intuitively, the new operator requires its argument to be satisfied within a bounded number of steps. 

The formulas of~$\prompt$ are given by the grammar
\[\varphi \cceq p \mid \neg p \mid \varphi \wedge \varphi \mid \varphi \vee \varphi \mid \X \varphi \mid \varphi \U \varphi  \mid \varphi \R \varphi 
 \mid \Diamondprompt \varphi
\]
where $p$ ranges over the atomic propositions in $P$. The size~$\size{\varphi}$ of a formula~$\varphi$ is defined as the number of its distinct subformulas.

In \prompt, formulas are in negation normal form and implication is disallowed.
Both requirements are necessary to preserve monotonicity of the prompt-eventually~$\Diamondprompt$ with respect to the parameter~$k$ bounding it (Alur et al.~\cite{AlurEtessamiLaTorrePeled01} provide a detailed discussion).
We define the shorthands~$\Diamond \varphi = \ttrue \U \varphi$ and $\Box \varphi = \ffalse \R \varphi$. 

Again, we define the semantics by an evaluation function~$\prompteval$ mapping a trace~$w \in (\pow{P})^\omega$, a bound~$k \in \nats$ for the prompt operators, and a formula~$\varphi$ to a truth value in $\bool$ (which is again equivalent to the standard definition).
This function is defined as usual for all Boolean and standard temporal operators (ignoring the bound~$k$), while a formula~$\Diamondprompt \varphi$ is satisfied with respect to the bound~$k$ if $\varphi$ holds within the next $k$ steps, i.e., the prompt-eventually behaves like the classical eventually with a bounded scope~\cite{KupfermanPitermanVardi09}:

\begin{itemize}
	\item $\prompteval(w , k , p) = \begin{cases}
 1 &\text{if $p \in w(0)$,}\\
 0 &\text{if $p \notin w(0)$,}	
 \end{cases}$
\quad\quad $\text{\normalfont\bfseries --} \,\,\, \prompteval(w , k , \neg p) = \begin{cases}
 1 &\text{if $p \notin w(0)$,}\\	
 0 &\text{if $p \in w(0)$,}
 \end{cases}
$
	\item $\prompteval(w , k , \varphi_0 \wedge \varphi_1) = \min \set{\prompteval(w, k , \varphi_0), \prompteval(w, k , \varphi_1) }$,
	
	\item $\prompteval(w , k , \varphi_0 \vee \varphi_1) = \max \set{\prompteval(w, k , \varphi_0), \prompteval(w, k , \varphi_1) }$,
	
	\item $\prompteval(w, k, \X \varphi) = \prompteval(\suff{w}{1}, k, \varphi)$,
	
	\item $\prompteval(w , k , \varphi_0 \U \varphi_1) = \max_{j \in \nats} \min \set{ \prompteval(\suff{w}{j}, k, \varphi_1 )  ,  \min_{0 \le j' < j} \prompteval(\suff{w}{j'}), k ,\varphi_0   }$,
	
	\item $\prompteval(w , k , \varphi_0 \R \varphi_1) = \min_{j \in \nats} \max \set{ \prompteval(\suff{w}{j}, k, \varphi_1 )  ,  \max_{0 \le j' < j} \prompteval(\suff{w}{j'}), k ,\varphi_0   }$,
	
	\item $\prompteval(w , k ,  \Diamondprompt \varphi) = \max_{0 \le j \le k} \prompteval(\suff{w}{j} , k , \varphi)$.

\end{itemize}

In verification problems for $\prompt$, the bound~$k$ on the prompt-eventuallies is existentially quantified.
Kupferman et al.\ proved that $\prompt$ model checking is $\pspace$-complete and that solving games with $\prompt$ winning conditions is $\twoexp$-complete~\cite{KupfermanPitermanVardi09}.\footnote{Instead of games, they actually considered the related framework of realizability, an abstract type of game without underlying graph. However, realizability and graph-based games are interreducible (also, see~\cite{Zimmermann13}).}

%% file: content/prel-promptldl.tex
In our proofs, we also use $\promptldl$, which can be seen as a combination of $\ldl$ and $\prompt$. This logic has been studied by Faymonville and Zimmermann~\cite{FaymonvilleZimmermann17} as a fragment of Parametric $\ltl$~\cite{AlurEtessamiLaTorrePeled01} (although  the logic has never been explicitly named). 

The formulas of $\promptldl$ are given by the grammar
\begin{align*}
\varphi &\cceq p \mid \neg p \mid \varphi \wedge \varphi \mid \varphi \vee \varphi \mid \ddiamond{r} \varphi 
  \mid \bbox{r} \varphi \mid \promptddiamond{r} \varphi \\
    r & \cceq \phi \mid \varphi? \mid r+r \mid r \conc r \mid r^*
\end{align*}
where $p$ again ranges over the atomic propositions in $P$ and $\phi$ ranges over propositional formulas over $P$. As in $\prompt$, we have to disallow arbitrary negations and implications. The size of a formula is defined as for $\ldl$. 

Furthermore, the semantics of $\promptldl$ is obtained by combining the one of $\ldl$ and the one of $\prompt$: Again, we define an evaluation function~$\promptldleval$ mapping a trace~$w$, a bound~$k$, and a formula~$\varphi$ to a truth value. 
\begin{itemize}
	\item $\promptldleval(w,k,p) = \begin{cases}
		1 &\text{if $p \in w(0)$,}\\
		0 &\text{if $p \notin w(0)$,}
	\end{cases}$\quad\quad $\text{\normalfont\bfseries --} \,\,\,\promptldleval(w,k,\neg p) = \begin{cases}
		1 &\text{if $p \notin w(0)$,}\\
		0 &\text{if $p \in w(0)$,}
	\end{cases}$
	\item $\promptldleval(w,k,\varphi_0 \wedge \varphi_1) = \min\set{\promptldleval(w,k, \varphi_0), \promptldleval(w,k, \varphi_1)}$,
	\item $\promptldleval(w,k,\varphi_0 \vee \varphi_1) = \max\set{\promptldleval(w,k, \varphi_0), \promptldleval(w,k, \varphi_1) }$,
	\item $\promptldleval(w,k,\ddiamond{r}\varphi) = \max_{j \in \Rexp(w,k,r)} \promptldleval(\suff{w}{j},k, \varphi)$, 
	\item $\promptldleval(w,k,\bbox{r}\varphi) = \min_{j \in \Rexp(w,k,r)} \promptldleval(\suff{w}{j},k, \varphi)$, and
	\item $\promptldleval(w,k,\promptddiamond{r}\varphi) = \max_{j \in \Rexp(w,k,r) \cap \set{0, \ldots, k}} \promptldleval(\suff{w}{j},k, \varphi)$.
\end{itemize}

Here, $\Rexp(w,k,r)$ is defined as $\Rexp(w,r)$, but propagates the bound~$k$ to evaluate tests. Hence, we define  $\Rexp(w,k,\varphi?) = \set{0}$ if $\ldleval(w,k, \varphi)=1$ and $\Rexp(w,k,\varphi?) =  \emptyset$ otherwise. All other cases are defined as for $\ldl$, but propagate the bound~$k$.

$\promptldl$ as defined here is a syntactic fragment of Parametric $\ldl$~\cite{FaymonvilleZimmermann17} and subsumes $\ltl$. Hence, its model checking problem is $\pspace$-complete and the synthesis problem is $\twoexp$-complete. Here, the bound~$k$ is again uniformly existentially quantified in verification problems.

%% file: content/rprompt.tex
We begin our treatment of combinations of the three basic logics by introducing robust semantics for $\prompt$, obtaining the logic~$\rprompt$. To this end, we add the prompt-eventually operator to $\rltl$ while disallowing implications and restricting negation to retain decidability (cf.~\cite{AlurEtessamiLaTorrePeled01}). The formulas of $\rprompt$ are given by
\[\varphi \cceq p \mid \neg p \mid \varphi \wedge \varphi \mid \varphi \vee \varphi \mid \Diamonddot \varphi 
  \mid \Boxdot \varphi\mid \Diamondpromptdot \varphi,
\]%
where $p$ ranges over the set $P$ of atomic propositions. 
The size~$\size{\varphi}$ of a formula~$\varphi$ is the number of its distinct subformulas. 

The semantics of $\rprompt$ is given by an evaluation function~$\rprompteval$ mapping a trace~$w$, a bound~$k$ for the prompt-eventuallies, and a formula~$\varphi$ to a truth value in $\bool_4$.  To simplify our notation, we write $\rprompteval_i(w,k,\varphi)$ for $i \in \set{1,2,3,4}$ to denote the $i$-th bit of $\rprompteval(w,k,\varphi)$, i.e., 
\[
\rprompteval(w,k,\varphi) = \rprompteval_1(w,k,\varphi)\rprompteval_2(w,k,\varphi)\rprompteval_3(w,k,\varphi)\rprompteval_4(w,k,\varphi).
\]
The semantics of Boolean connectives as well as of the eventually and always operators is defined as for robust $\ltl$. The motivation behind these definitions is carefully and convincingly discussed by Tabuada and Neider~\cite{TabuadaNeider16}. The semantics of the prompt-eventually operator bounds its scope to the next $k$ positions as in classical $\prompt$~\cite{KupfermanPitermanVardi09}.
\begin{itemize}
\item $\rldleval(w,k,p) =\begin{cases}
	1111 &\text{if $p \in w(0)$,}\\
	0000 &\text{if $p \notin w(0)$},
\end{cases}$ \quad\quad $\text{\normalfont\bfseries --} \,\,\, \rldleval(w,k,\neg p) =\begin{cases}
	1111 &\text{if $p \notin w(0)$,}\\
	0000 &\text{if $p \in w(0)$},
 \end{cases}$

\item
$\rldleval(w, k,\varphi_0 \wedge \varphi_1) = \min\set{\rldleval(w,k, \varphi_0), \rldleval(w,k,\varphi_1) }$, 
\item $\rldleval(w, k,\varphi_0 \vee \varphi_1) = \max\set{\rldleval(w, k,\varphi_0), \rldleval(w,k,\varphi_1) }$, 
	\item $\rprompteval(w,k,\Diamonddot \varphi) = b_1 b_2 b_3 b_4$ where $b_i = \max_{j \in \nats}\rprompteval_i(\suff{w}{j},k,\varphi )$,\footnote{This definition is equivalent to $\rprompteval(w,k,\Diamonddot \varphi) = \max_{j \in \nats}\rprompteval(\suff{w}{j},k,\varphi )$ due to monotonicity of the truth values, which is closer to the classical semantics of the eventually operator. A similar equivalence holds for $\Diamondprompt\varphi$.} and
	\item $\rprompteval(w,k, \Boxdot\varphi) = b_1 b_2 b_3 b_4$ where
\begin{itemize}
	
	\item $b_1 = \min_{j \in \nats} \rprompteval_1(\suff{w}{j},k,\varphi )$, i.e., $b_1 = 1$ iff $\varphi$ holds always,
	 
	\item $b_2 = \max_{j' \in \nats} \min_{j' \le j} \rprompteval_2(\suff{w}{j},k,\varphi )$, i.e., $b_2 = 1$ iff $\varphi$ holds almost always, 
		
	\item $b_3 = \min_{j' \in \nats} \max_{j' \le j} \rprompteval_3(\suff{w}{j},k,\varphi )$, i.e., $b_3 = 1$ iff $\varphi$ holds infinitely often, and 
	
	\item $b_4 = \max_{j \in \nats} \rprompteval_4(\suff{w}{j},k,\varphi )$ i.e., $b_4 = 1$ iff $\varphi$ holds at least once.

\end{itemize}
\item $\rprompteval(w,k, \Diamondpromptdot \varphi) = b_1 b_2 b_3 b_4$ where $b_i = \max_{0 \le j \le k}\rprompteval_i(\suff{w}{j},k,\varphi )$.
\end{itemize}
It is easy to verify that $\rprompteval(w,k,\varphi)$ is well-defined, i.e., $\rprompteval(w,k,\varphi) \in \bool_4$ for all $w$, $k$, and $\varphi$. 

\begin{example}
\label{example-rprompt}
Consider the formula~$\varphi = \Boxdot \Diamondpromptdot s$, where we interpret occurrences of the atomic proposition~$s$ as synchronizations. Then, the different degrees of satisfaction of the formula express the following possibilities, when evaluating it with respect to $k \in \nats$: (i) the distance between synchronizations is bounded by $k$, (ii) from some point onwards, the distance between synchronizations is bounded by $k$, (iii) there are infinitely many synchronizations, and (iv) there is at least one synchronization. Note that the last two possibilities are independent of $k$, which is explained by simple logical equivalences, e.g., the third possibility reads actually as follows: there are infinitely many positions such that a synchronization occurs within distance~$k$. However, it is easy to see that is equivalent to the property stated above. 
\end{example}

In the next two sections, we solve the model checking problem and the synthesis problem for $\rprompt$. To this end, we translate every $\rprompt$ formula into a sequence of five $\prompt$ formulas that capture the five degrees of satisfaction and violation by making the semantics of the robust always operator explicit. This is a straightforward generalization of the, in the terms of the introduction, reduction-based approach to robust $\ltl$~\cite{TabuadaNeider16}. 

\begin{lemma}
\label{lemma-prltl2pltl}
For every $\rprompt$ formula~$\varphi$ and every $\tval \in \bool_{4}$, there is a $\prompt$ formula~$\varphi_\tval$ of size~$\bigo(\size{\varphi})$ such that $\rprompteval(w, k, \varphi) \succeq \tval$ if and only if $\prompteval(w, k, \varphi_\tval) = 1$.
\end{lemma}

\begin{proof}
If $\tval = 0000$, then we can pick $\varphi_\tval = \ttrue$, independently of $\varphi$. Otherwise, we obtain the result by induction over the construction of $\varphi$ implementing the intuition behind the robust semantics:
\begin{itemize}

	\item $p_\tval = p$ and $(\neg p)_\tval = \neg p$ for all atomic propositions~$p \in P$ and all $\tval \succ 0000$. 
	
	\item $(\varphi_0 \wedge \varphi_1)_\tval = (\varphi_0)_\tval \wedge (\varphi_1)_\tval$ for all $\tval \succ 0000$. 

	\item $(\varphi_0 \vee \varphi_1)_\tval = (\varphi_0)_\tval \vee (\varphi_1)_\tval$ for all $\tval \succ 0000$. 
	
	\item $(\Diamonddot\varphi )_\tval = \Diamond(\varphi_\tval)$ for all $\tval \succ 0000$.
	
	\item 

	$(\Boxdot \varphi)_{1111} = \Box (\varphi_{1111})$.
\item
	$(\Boxdot \varphi)_{0111} = \Diamond\Box (\varphi_{0111})$.
\item
	$(\Boxdot \varphi)_{0011} = \Box\Diamond (\varphi_{0011})$. 
\item
	$(\Boxdot \varphi)_{0001} = \Diamond (\varphi_{0001})$.
	
	\item $(\Diamondpromptdot \varphi )_\tval = \Diamondprompt (\varphi_\tval)$ for all $\tval \succ 0000$.

\end{itemize}
A straightforward induction shows that the resulting formula has the desired properties.
\end{proof}

Note that the logic $\rltl$ is not a fragment of $\rprompt$ as we have to disallow negation and implication to retain decidability~\cite{AlurEtessamiLaTorrePeled01}.
Conversely, $\prompt$ is also not a fragment of $\rprompt$ as we omitted the next, until, and release operator. However, we present a reduction-based approach from $\rprompt$ to $\prompt$. Thus, one could easily add the additional temporal operators to $\rprompt$ while maintaining the result of Lemma~\ref{lemma-prltl2pltl}. We prefer not to do so for the sake of accessibility and brevity.

%% file: content/rprompt-mc.tex
Let us now consider the $\rprompt$ model checking problem, which asks whether all executions of a given finite transition system satisfy a given specification expressed as an $\rprompt$ formula with truth value at least $\tval \in \bool_4$.
More formally, we assume the system under consideration to be modeled as a (labeled and initialized) transition system~$\sys = (S, s_\init, E, \lambda)$ over~$P$ consisting of a finite set~$S$ of states containing the initial state~$s_\init$, a directed edge relation~$E \subseteq S \times S$, and a state labeling~$\lambda \colon S \rightarrow 2^P$ that maps each state to the set of atomic propositions that hold true in this state. A path through $\sys$ is a sequence~$\rho = s_0 s_1 s_2 \cdots $ satisfying $s_0 = s_\init$ and $(s_j, s_{j+1}) \in E$ for every $j \in \nats$, and $\Pi_\sys$ denotes the set of all paths through $\sys$. Finally, the trace of a path~$\rho = s_0 s_1 s_2 \cdots \in \Pi_\sys$ is the sequence $\trace(\rho) = \lambda(s_0) \lambda(s_1) \lambda(s_2)\cdots$ of labels induced by $\rho$.
\begin{problem} \label{prob:prLTL-model-checking}
Let $\varphi$ be an $\rprompt$ formula, $\sys$ a transition system, and $\tval \in \bool_4$. Is there a $k \in\nats$ such that $\rprompteval(\trace(\rho), k, \varphi) \succeq \tval$ holds true for all paths~$\rho \in \Pi_\sys$?
\end{problem}

Our solution relies on Lemma~\ref{lemma-prltl2pltl} and on $\prompt$ model checking being in $\pspace$~\cite{KupfermanPitermanVardi09}.

\begin{theorem}
\label{thm-rpltlmodelchecking}
$\rprompt$ model checking is in $\pspace$.
\end{theorem}

\begin{proof}
By Lemma~\ref{lemma-prltl2pltl}, there exists a $k \in\nats$ such that $\rprompteval(\trace(\rho), k, \varphi) \succeq \tval$ holds true for all paths~$\rho \in \Pi_\sys$ if and only if there exists a $k \in \nats$ such that $\prompteval(\trace(\rho), k, \varphi_\tval) =1$ for all paths~$\rho \in \Pi_\sys$. The latter is an instance of the $\prompt$ model checking problem, which is in $\pspace$~\cite{KupfermanPitermanVardi09}.
\end{proof}

We do not claim $\pspace$-hardness because model checking the fragment of $\ltl$ with disjunction, conjunction, always, and eventually operators only (and classical semantics) is $\np$-complete~\cite{AlurLatorre04}. Since this fragment can be embedded into $\rprompt$ (via a translation of this $\ltl$ fragment into $\rprompt$ using techniques similar to those presented by Tabuada and Neider~\cite{TabuadaNeider16} for translating $\ltl$ into $\rltl$), we obtain at least $\np$-hardness for Problem~\ref{prob:prLTL-model-checking}. As we have no next, until, and release operators (by our own volition), we cannot easily claim $\pspace$-hardness. In contrast, the solution of the $\prompt$ model checking problem consists of a reduction to $\ltl$ model checking that introduces until operators (see~\cite{KupfermanPitermanVardi09}). Hence, we leave the fragment mentioned above, for which $\np$ membership is known. However, adding next, until, and release to $\rprompt$ yields a $\pspace$-hard model checking problem. 

%% file: content/rprompt-synt.tex
Next, we consider the problem of synthesizing reactive controllers from $\rprompt$ specifications. In this context, we rely on the classical reduction from reactive synthesis to infinite-duration two-player games over finite graphs. In particular, we show how to construct a finite-state winning strategy for games with $\rprompt$ winning conditions, which immediately correspond to implementations of reactive controllers.
Throughout this section, we assume familiarity with games over finite graphs~(see, e.g., \cite[Chapter~2]{GraedelThomasWilke02}).

We consider $\rprompt$ games over~$P$, which are triples $\game = (\ggraph, \varphi, \tval)$ consisting of
	 a labeled game graph $\ggraph$,
an $\rprompt$ formula $\varphi$, and a truth value $\tval \in \bool_4$.
A labeled game graph~$\ggraph = (V_0, V_1, E, \lambda)$ consists of a  directed graph $(V_0 \cup V_1, E)$, two finite, disjoint sets of vertices~$V_0$ and~$V_1$, and a function $\lambda \colon V_0 \cup V_1 \to 2^P$ mapping each vertex~$v$ to the set~$\lambda(v)$ of atomic propositions that hold true in $v$.
We denote the set of all vertices by $V = V_0 \cup V_1$ and assume that game graphs do not have terminal vertices, i.e., $\{ v \} \times V \cap E \neq \emptyset$ for each $v \in V$.

As in the classical setting, $\rprompt$ games are played by two players, Player~0 and Player~1, who move a token along the edges of the game graph ad infinitum (if the token is currently placed on a vertex $v \in V_i$, $i \in \{ 0, 1\}$, then Player~$i$ decides the next move). 
The resulting infinite sequence $\rho = v_0 v_1 v_2\cdots \in V^\omega$ of vertices is called a {play} and induces a trace $\lambda(\rho) = \lambda(v_0) \lambda(v_1)\lambda(v_2) \cdots \in (2^P)^\omega$.

A strategy of Player~$0$ is a mapping $f \colon V^\ast V_0 \to V$ that prescribes where to move the token depending on the finite play prefix constructed so far. A play~$v_0v_1v_2 \cdots$ is played according to $f$ if $v_{j+1} = f(v_0 \cdots v_j)$ for every $j$ with $v_j \in V_0$. A strategy $f$ of Player~$0$ is winning from a vertex~$v \in V$ if there is a $k \in \nats$ such that all plays~$\rho$ that start in $v$ and that are played according to $f$ satisfy $\rprompteval(\lambda(\rho), k, \varphi) \succeq \tval$, i.e., the evaluation of $\varphi$ with respect to $k$ on $\lambda(\rho)$ determines the winner of the play~$\rho$. Further, a (winning) strategy is a finite-state strategy if there exists a finite-state machine computing it in the usual sense (see~\cite[Chapter~2]{GraedelThomasWilke02} for details).

We are interested in solving $\rprompt$ games, i.e., in solving the following problem.

\begin{problem} 
Let $\game$ be an $\rprompt$ game and $v$ a vertex. Determine whether Player~$0$ has a winning strategy for $\game$ from $v$ and compute a finite-state winning strategy if so.
\end{problem}

Again, our solution to this problem relies on Lemma~\ref{lemma-prltl2pltl} and the fact that solving $\prompt$ games is in $\twoexp$~\cite{KupfermanPitermanVardi09,Zimmermann13}.

\begin{theorem} \label{thm-rpltlgames}
Solving $\rprompt$ games is $\twoexp$-complete. 	
\end{theorem}

\begin{proof}
The lower bound follows from the special case of $\ltl(\Box, \Diamond)$~\cite{gameswithboxes}, which is a fragment of $\rprompt$. On the other hand, the upper bound is again proven by a reduction to $\prompt$: Player~$0$ having a winning strategy for $(\ggraph, \varphi, \tval)$ from $v$ is equivalent to her having a winning strategy for the $\prompt$ game~$(\ggraph, \varphi_\tval)$ from $v$. The latter problem can be solved in doubly-exponential time and a finite-state strategy can effectively be computed~\cite{Zimmermann13}.
\end{proof}

Here we have a matching lower bound, as solving games with $\ltl$ conditions without next, until, and release is already $\twoexp$-hard~\cite{gameswithboxes}.

%% file: content/rldl.tex
Next, we \myquot{robustify} $\ldl$ by generalizing the ideas underlying robust $\ltl$ to $\ldl$, obtaining the logic $\rldl$.  Again, following the precedent of robust $\ltl$, we equip robust operators with dots to distinguish them from non-robust ones.
The formulas of $\rldl$ are given by the grammar
\begin{align*}
\varphi &\cceq p \mid \neg \varphi \mid \varphi \wedge \varphi \mid \varphi \vee \varphi \mid \varphi \implies \varphi \mid  \ddiamonddot{r} \varphi 
  \mid \bboxdot{r} \varphi \\
    r  &\cceq \phi \mid \varphi? \mid r+r \mid r \conc r \mid r^*
\end{align*}
where $p$ ranges over the atomic propositions in $P$ and $\phi$ over propositional formulas over $P$. We refer to formulas of the form~$\ddiamonddot{r}\varphi$ and $\bboxdot{r}\varphi$ as diamond formulas and box formulas, respectively. In both cases, $r$ is the guard of the operator. An atom~$\varphi?$ of a regular expression is a test. We use the abbreviations~$\ttrue = p \vee \neg p$ and $\ffalse = p \wedge \neg p$ for some $p \in P$ and note that both are formulas and guards.
We denote the set of subformulas of $\varphi$ by $\cl(\varphi)$. Guards are not subformulas, but the formulas appearing in the tests are, e.g., we have \[\cl(\ddiamonddot{\halfthinspace p?\conc q}p') = \set{\halfthinspace p, p', \ddiamonddot{\halfthinspace p?\conc q}p'}.\] 
The size~$\card{\varphi}$ of $\varphi$ is the sum of $\card{\cl(\varphi)}$ and the sum of the lengths of the guards appearing in $\varphi$ (measured in the number of operators),  taking each occurrence of a guard in the syntax tree (not the syntax DAG like for subformulas) into account.

Before we introduce the semantics of $\rldl$ we first recall the semantics of the robust always operator~$\Boxdot\varphi$ in robust $\ltl$. To this end, call a position~$j$ of a trace $\varphi$-satisfying if the suffix starting at position~$j$ satisfies $\varphi$. Now, the robust semantics are based on the following five cases, where the latter four distinguish various degrees of violating the formula~$\Box\varphi$: either all positions are $\varphi$-satisfying ($\Box$), almost all positions are $\varphi$-satisfying ($\Diamond\Box$), infinitely many positions are $\varphi$-satisfying ($\Box\Diamond$), some position is $\varphi$-satisfying ($\Diamond$), or no position is $\varphi$-satisfying.

A similar approach for a  formula~$\bboxdot{r}\varphi$ would be to consider the following possibilities, where a position~$j$ of a trace~$w$ is an $r$-match if the prefix of $w$ up to and including position $j-1$ is in the language of $r$: all $r$-matches are $\varphi$-satisfying, almost all $r$-matches are $\varphi$-satisfying, infinitely many $r$-matches are $\varphi$-satisfying, some $r$-match is $\varphi$-satisfying, or no $r$-match is $\varphi$-satisfying.
On a trace~$w$ with infinitely many $r$-matches, this is the natural generalization of the robust semantics. A trace, however, may only contain finitely many $r$-matches, or none at all. In the former case, there are not infinitely many $\varphi$-satisfying $r$-matches, but all $r$-matches could satisfy $\varphi$. Thus, the monotonicity of the cases is violated. 
We overcome this by interpreting \myquot{almost all} as \myquot{all} and \myquot{infinitely many} as \myquot{some} if there are only finitely many $r$-matches.\footnote{\label{footnote-altsemantics}There is an alternative definition inspired by the semantics of $\ltl$ on finite traces: Here, both $\Diamond\Box\varphi$ and $\Box\Diamond\varphi$ are equivalent to \myquot{$\varphi$ holds at the last position}. This suggests interpreting \myquot{almost all $r$-matches are $\varphi$-satisfying} and \myquot{infinitely many $r$-matches are $\varphi$-satisfying} as \myquot{the last $r$-match is $\varphi$-satisfying} in case there are only finitely many $r$-matches. Arguably, this definition is less intuitive than the one we propose to pursue.}

Also, the guard~$r$ may contain tests, which have to be evaluated to determine whether a position is an $r$-match. For this, we have to use the appropriate semantics for the robust box operator. For example, if we interpret $\bboxdot{r}\varphi$ to mean \myquot{almost all $r$-matches satisfy $\varphi$}, then the robust box operators in tests of $r$ are evaluated with this interpretation as well. This may, however, violate monotonicity (see Example~\ref{example-monotonicityviolation}), which we therefore hardcode in the semantics.

We now formalize the informal description above and subsequently show that this formalization satisfies all desired properties. To this end, we again define an evaluation function~$\rldleval$ mapping a trace~$w$ and a formula~$\varphi$ to a truth value. Also, we again denote the projection of $\rldleval(w,\varphi)$ to its $i$-th bit by $\rldleval_i(w,\varphi)$. For atomic propositions, conjunction, disjunction, negation, and implication, the definition is the same as for robust $\ltl$ on  Page~\pageref{rltlsemdef}.

To define the semantics of the diamond and the box operator, we need to first define the semantics of the guards: The match set~$\robRexp_i(w,r) \subseteq \nats$ for $i \in \set{1,2,3,4}$ contains all positions~$j$ of $w$ such that $\pref{w}{j}$ matches $r$ (with tests in $r$ being evaluated depending on the value of $i$) and is defined inductively as follows: 
\begin{itemize}

\item $\robRexp_i(w,\phi) = \set{1}$ if $w(0) \models \phi$ and $\robRexp_i(\phi,w) =  \emptyset$ otherwise, for propositional~$\phi$.

\item $\robRexp_i(w,\varphi?) = \set{0}$ if $\rldleval_i(w, \varphi)=1$ and $\robRexp_i(w,\varphi?) =  \emptyset$ otherwise.

\item $\robRexp_i(w,r_0 + r_1) = \robRexp_i(w,r_0) \cup \robRexp_i(w,r_1)$.

\item $\robRexp_i(w,r_0 \conc r_1) = \set{j_0 + j_1 \mid j_0, j_1 \ge 0 \text{ and } j_0 \in \robRexp_i(w,r_0) \text{ and } j_1 \in \robRexp_i(\suff{w}{j_0}, r_1)}$, i.e., for $j$ to be in $\robRexp_i(w,r_0 \conc r_1)$, it has to be the sum of natural numbers $j_0$ and $j_1$ such that $w$ has a prefix of length $j_0$ that matches $r_0$ and $\suff{w}{j_0}$ has a prefix of length~$j_1$ that matches~$r_1$ (where in both cases the tests are again evaluated depending on the value of $i$). 

\item $\robRexp_i(w,r^*) = \set{0} \cup \set{j_1 + \cdots + j_\ell \mid 0 \le j_{\ell'} \in \robRexp_i(\suff{w}{j_1 + \cdots + j_{\ell'-1}},r) \text{ for all } \ell'\in\set{1, \ldots, \ell}}$, where we use $j_1 + \cdots + j_{0} = 0$. Thus, for $j$ to be in $\robRexp_i(w,r^*)$, it has to be expressible as $j = j_1 +\cdots + j_\ell$ with non-negative~$j_{\ell'}$ such that the prefix of $w$ of length~$j_1$ matches $r$, the prefix of length~$j_2$ of $\suff{w}{j_1}$ matches $r$, and in general, the prefix of length~$j_{\ell'}$ of $\suff{w}{j_1+\cdots +j_{\ell'-1}}$ matches $r$, for every $\ell' \in \set{1, \ldots, \ell}$ (where the tests are evaluated depending on $i$).

\end{itemize}
Due to tests, membership of $j$ in $\robRexp_i(w,r)$ does, in general, not only depend on the prefix~$\pref{w}{j}$, but on the complete trace~$w$.
Also, the semantics of the propositional atom~$\phi$ differs from the semantics of the test~$\phi?$: the former consumes an input letter, while the latter one does not.
Thus, $\rldl$ (as $\ldl$) features both kinds of atoms. 
We define the intuition given above via  
\begin{itemize}
\item $\rldleval(w,  \ddiamonddot{r}\varphi) = b_1 b_2 b_3 b_4$ where 
$b_i = \max\nolimits_{j \in \robRexp_i(w,r)} \rldleval_i(\suff{w}{j}, \varphi)$ and
\item $\rldleval(w, \bboxdot{r}\varphi) = b_1 b_2 b_3 b_4$ with $b_i = \max\set{b_1', \ldots, b_i'}$ for every $i \in \set{ 1,2,3,4}$, where
\begin{itemize}
	
	\item $b_1' = \min_{j \in \robRexp_1(w,r)} \rldleval_1(\suff{w}{j},\varphi)$,
	
	\item $b_2' = \begin{cases}
	\max_{j' \in \nats} \min_{j \in \robRexp_2(w,r) \cap \set{j', j'+1, j'+2, \ldots}} \rldleval_2(\suff{w}{j},\varphi)&\text{if $\size{\robRexp_2(w,r)} = \infty$},\\
	
	\min_{j \in \robRexp_2(w,r)} \rldleval_2(\suff{w}{j},\varphi)&\text{if $0 < \size{\robRexp_2(w,r)} < \infty$},\\
	
	1 &\text{if $\size{\robRexp_2(w,r)} = 0$},
	\end{cases}$
	
	\item $b_3' = \begin{cases}
	\min_{j' \in \nats} \max_{j \in \robRexp_3(w,r) \cap \set{j', j'+1, j'+2, \ldots}} \rldleval_3(\suff{w}{j},\varphi)&\text{if $\size{\robRexp_3(w,r)} = \infty$},\\
	
	\max_{j \in \robRexp_3(w,r)} \rldleval_3(\suff{w}{j},\varphi)&\text{if $0 < \size{\robRexp_3(w,r)} < \infty$},\\
	
	1 &\text{if $\size{\robRexp_3(w,r)} = 0$},
	\end{cases}$
	
	\item $b_4' =\begin{cases} \max_{j \in \robRexp_4(w,r)} \rldleval_4(\suff{w}{j},\varphi) &\text{if $\size{\robRexp_4(w,r)} > 0$,}\\
1&\text{if $\size{\robRexp_4(w,r)} = 0$.}
\end{cases}
$
\end{itemize}
\end{itemize}

To give an intuitive description of the semantics, let us first generalize the notion of $r$-matches and $\varphi$-satisfiability.
We say that a position~$j$ of $w$ is an $r$-match of degree~$\tval$ if $j \in \robRexp_i(w,r)$ for the unique $i$ with $\tval = \itotruthvalue{i}$, which requires all tests in~$r$ to be evaluated w.r.t.\ $\rldleval_i$ (i.e., to some truth value at least $\tval$). Similarly, we say that a position~$j$ of $w$ is $\varphi$-satisfying of degree~$\tval$ if $\rldleval(\suff{w}{j} ,\varphi) \succeq \tval$, or if, equivalently, $\rldleval_i(\suff{w}{j},\varphi) =1$ for the unique $i$ with $\tval = \itotruthvalue{i}$.

Now, consider the $b_i'$ defining the semantics of the robust box operator: We have $b_1' = 1$ if all $r$-matches of degree~$1111$ are $\varphi$-satisfying of degree~$1111$. This is in particular satisfied if there is no such match.
Further, if there are infinitely (finitely) many $r$-matches of degree~$0111$, then $b_2' =1$ if almost all (if all) those matches are $\varphi$-satisfying of degree~$0111$. 
Dually, if there are infinitely (finitely) many $r$-matches of degree~$0011$, then $b_3' =1$ if infinitely many (at least one) of those matches are (is) $\varphi$-satisfying of degree~$0011$.
Finally, if there is at least one $r$-match of degree~$0001$, then $b_4' =1$ if at least one of those matches is $\varphi$-satisfying of degree~$0001$.
The cases where there is no $r$-match are irrelevant due to monotonicity, so we hardcode them to $1$.

\begin{example}
Consider the formula~$\bboxdot{r}q \implies \bboxdot{\ttrue\conc r}p$ with $r = (\ttrue; \ttrue)^*$, which expresses that the degree of violation of $q$ at \emph{even} positions should at most be the degree of violation of $p$ at \emph{odd} positions. Such a property cannot be expressed in $\rltl$, as even $\bboxdot{r}q$ is known to be inexpressible in $\ltl$~\cite{BaierKatoen08}.
\end{example}

First, we state that the semantics is well-defined. This is not obvious due to the case distinctions and the use of the matching sets~$\robRexp_i$ for different $i$.

\begin{lemma}
\label{lemma-semanticswelldefineds}
We have $\rldleval(w, \varphi) \in \bool_4$ for every trace~$w$ and every formula~$\varphi$.
\end{lemma}

 \begin{proof}
We proceed by induction over the structure of $\varphi$. The cases of atomic propositions and Boolean connectives are trivial, as they return values from $\bool_4$ by definition, provided the arguments are from $\bool_4$. Similarly, we have $\rldleval(w, \bboxdot{r}\varphi) = b_1 b_2 b_3 b_4 \in \bool_4$, as the maximization \myquot{$b_i = \max\set{b_1', \ldots, b_i'}$} in the definition  enforces the desired monotonicity of the bits~$b_i$. 

To conclude, consider a diamond formula~$ \ddiamonddot{r}\varphi$. Applying the induction hypothesis to the tests of $r$ and an induction over the construction of $r$ shows 
\[\robRexp_1(w,r) \subseteq \robRexp_2(w,r) \subseteq \robRexp_3(w,r) \subseteq \robRexp_4(w,r)  
\]
for every trace~$w$. Hence, an application of the induction hypothesis for $\varphi$ yields
\begin{align*}
\max_{j \in \robRexp_1(w,r)} \rldleval_1(\suff{w}{j}, \varphi)& \le 
\max_{j \in \robRexp_2(w,r)} \rldleval_2(\suff{w}{j}, \varphi) \\ 
&\le \max_{j \in \robRexp_3(w,r)} \rldleval_3(\suff{w}{j}, \varphi) \le
\max_{j \in \robRexp_4(w,r)} \rldleval_4(\suff{w}{j}, \varphi) 
\end{align*}
for every trace~$w$. Hence, $\rldleval(w,  \ddiamonddot{r}\varphi) \in \bool_4$.
\end{proof}

To conclude the definition of the semantics, we give an example witnessing that the maximization over the $b_i'$ in the semantics of the box operator is indeed necessary to obtain monotonicity.

\begin{example}
\label{example-monotonicityviolation}
Let $\varphi = \bboxdot{r}\ffalse$ with $r = (\bboxdot{\ttrue^*}p)?$.
Moreover, consider the trace~$w = \emptyset\set{p}^\omega$. 
Then, we have $\rldleval(w, \bboxdot{\ttrue^*}p) = 0111$ and consequently $\robRexp_1(w,r) = \emptyset$ and $\robRexp_2(w,r) = \set{0}$. Therefore, $\min_{j \in \robRexp_1(w,r)} \rldleval_1(\suff{w}{j},\ffalse) = \min \emptyset = 1$, but $\min_{j \in \robRexp_2(w,r)} \rldleval_2(\suff{w}{j},\ffalse) = \min \set{0} = 0$. Thus, the bits~$b_1'$ and $b_2'$ inducing  $\rldleval(w, \bboxdot{r}\ffalse)$ are not monotonic, which explains the need to maximize  over the $b_i'$ to obtain the semantics of the robust box operator.
The traces~$(\emptyset \set{p})^\omega$ and~$\set{p} \emptyset^\omega$ witness that monotonicity can also be violated for the pairs $b_2',b_3'$ and $b_3',b_4'$.
\end{example}

We prove that $\rldl$ has the exponential compilation property. This allows us to solve the model checking and the synthesis problem using well-known and efficient automata-based algorithms. Furthermore, we are able to show that the complexity of these algorithms is asymptotically the same as the complexity of the algorithms for plain $\ldl$ and $\ltl$. 
In the terminology introduced in the introduction, we present a direct translation, i.e., we translate $\rldl$ directly into automata. 

\begin{theorem}
\label{theorem-translation-oldcor}
Let $\varphi$ be an $\rldl$ formula, $n = \size{\varphi}$, and $\tval \in \bool_4$.
There is a non-deterministic Büchi automaton~$\autb_{\varphi, \tval}$ with $2^{\bigo(n )}$ states recognizing the language~$\set{w \in (\pow{P})^\omega \mid \rldleval(w,\varphi) \succeq \tval}$.

\end{theorem}

In order to prove this theorem, we first recall in Section~\ref{subsec-markedautomata} how to translate guards~$r$ into finite non-deterministic automata with special features to account for tests.
Then, in Section~\ref{subsec-alternatingautomata}, we present the translation of $\rldl$ into weak alternating Büchi automata of linear size, which can then be further transformed into non-deterministic Büchi automata of exponential size and deterministic parity automata of doubly-exponential size.
Such automata are needed for solving the model checking problem and the synthesis problem, respectively.

\subsection{Translating Guards into Automata}
\label{subsec-markedautomata}

Recall that $P$ is the (finite) set of atomic propositions.
An automaton with tests~$\autr = (Q, \pow{P}, q_\init, \delta, F, \marking)$ consists of a finite set~$Q$ of states, the alphabet~$\pow{P}$, an initial state~$q_\init \in Q$, a transition function~$\delta \colon Q \times (\pow{P}\cup\set{\epsilon}) \rightarrow \pow{Q}$, a set~$F$ of final states, and a partial function~$\marking$, which assigns to states~$q \in Q$ an $\rldl$ formula~$\marking(q)$. These should be thought of as the analogue of tests, i.e., if $\marking(q)$ is defined, then a run visiting $q$ is only successful if the word that remains to be processed from $q$ onwards satisfies the formula $\marking(q)$.

We write $q \xrightarrow{a} q'$ if $q' \in \delta(q, a)$ for $a \in \pow{P} \cup \set{\epsilon}$. An $\epsilon$-path~$\pi$ from $q$ to $q'$ in $\autr$ is a sequence~$\pi = q_1 \cdots q_k$ of $k \ge 1$ states with $q =q_1 \xrightarrow{\epsilon} \cdots \xrightarrow{\epsilon} q_k = q'$.  Let $\marking(\pi) = \set{\marking(q_i) \mid 1 \le i \le k}$  denote the set of tests visited by $\pi$ and let $\Pi(q, q')$ denote the set of all $\epsilon$-paths from $q$ to $q'$.

A run of $\autr$ on $w(0) \cdots w(j-1) \in (\pow{P})^*$ is a sequence~$q_0 q_1 \cdots q_j$ of states such that $q_0 = q_\init$ and for every $j'$ in the range~$0 \le j' \le j-1$ there is a state~$q_{j'}'$ reachable from $q_j$ via an $\epsilon$-path~$\pi_{j'}$ and such that $q_{j'+1} \in \delta(q_{j'}', w(j'))$. The run is accepting if there is a $q_{j}' \in F$ reachable from $q_j$ via an $\epsilon$-path~$\pi_j$.  This slightly unusual definition of runs (but equivalent to the standard one) simplifies our reasoning below. Also, the definition is oblivious to the tests assigned by~$\marking$. To take them into account, we define for $i \in \set{1,2,3,4}$
\begin{multline*}
\robRexp_i(w,\autr) = \{j \mid \text{$\autr$ has an accepting run on $\pref{w}{j}$ with $\epsilon$-paths $\pi_0, \ldots, \pi_j$ s.t.}\\
\text{$\rldleval_i(\suff{w}{j'}, \bigwedge \marking(\pi_{j'}))=1$ for every $j'$ in the range~$0 \le j' \le j$\}}.	
\end{multline*}
Here, $\bigwedge \marking(\pi_{j'})$ is the conjunction of all formulas in $\marking(\pi_{j'})$.

Every guard (which is just a regular expression with tests) can be turned into an equivalent automaton with tests via a straightforward generalization of the classical Thompson construction turning classical regular expressions into $\epsilon$-NFA, to which one adds a rule turning a test into a one-state automaton whose state is labeled with this test (see Figure~2 of Faymonville and Zimmermann~\cite{FaymonvilleZimmermann14} for details).

\begin{lemma}
\label{lemma-guards2automata}
Every guard~$r$ can be translated into an automaton with tests~$\autr_r$ such that $\robRexp_i(w,r) = \robRexp_i(w,\autr_r)$ for every $i \in 
\set{1,2,3,4}$ and with $\size{\autr_r} \in \bigo(\size{r})$. Furthermore, all final states of $\autr_r$ are terminal, i.e., they have no outgoing transitions.
\end{lemma}

The automaton~$\autr_r$ is independent of $i$, as this value only determines how tests are evaluated. These are handled \myquot{externally} in the definition of the semantics.
Having thus demonstrated how to turn guards into automata, we now demonstrate how to do the same for \rldl\ formulas.

\subsection{Translating rLDL into Alternating Automata}
\label{subsec-alternatingautomata}

In this subsection, we translate $\rldl$ formulas into weak alternating Büchi automata, which are known to be translatable into non-deterministic Büchi automata of exponential size~\cite{MH}. Hence, the linear translation from $\rldl$ to weak alternating Büchi automata we are about to present implies the exponential compilation property for $\rldl$. 

An alternating Büchi automaton~$\aut = (Q,\Sigma,q_\init,\delta, \accpar)$ consists of a finite set~$Q$ of states, an alphabet~$\Sigma$, an initial state~$q_\init \in Q$, a transition function~$\delta \colon Q \times \Sigma \to \bplus(Q)$, and a set~$\accpar \subseteq Q$ of accepting states. 
Here, $\bplus(Q)$ denotes the set of positive Boolean combinations over $Q$, which contains in particular the formulas $\ttrue$ (true) and $\ffalse$ (false).

A run of $\aut$ on $w = w(0) w(1) w(2) \cdots \in \Sigma^\omega$ is a directed graph~$\rho = (V, E)$ where $V \subseteq Q \times \nats$ and $((q,n),(q',n')) \in E$ implies $n' = n +1$, such that $(q_\init, 0) \in V$, and such that for all $(q, n) \in V$ we have $\suc{
\rho}{(q,n)} \models \delta(q, w(n))$. Here $\suc{\rho}{(q,n)}$ denotes the set of successors of $(q,n)$ in $\rho$ projected to $Q$. A run~$\rho$ is accepting if all infinite paths (projected to $Q$) through $\rho$ satisfy the Büchi condition, i.e., an accepting state is visited infinitely often on the path. The language~$L(\aut)$ contains all $w \in \Sigma^\omega$ that have an accepting run of $\aut$.

Given $\aut$ as above, its transition graph~$(Q,E)$ is the directed graph such that $(q,q') \in E$ if and only if $q'$ appears in $\delta(q,a)$ for some $a \in \Sigma$. The automaton~$\aut$ is weak, if every strongly connected component~$C$ of $(Q,E)$ is either a subset of $F$ or a subset of $Q \setminus F$, i.e., every cycle has either only accepting states or only rejecting states.

Using standard constructions, weak alternating Büchi automata are easily seen to be closed under all Boolean operations.
 Fix automata~$\aut_0 =  (Q_0,\Sigma,q_\init^0,\delta_0, \accpar_0) $ and $\aut_1= (Q_1,\Sigma,q_\init^1,\delta_1, \accpar_1)$.
\begin{itemize}
	
	\item $(Q_0, \Sigma, q_\init^0,  \overline{\delta_0}, \overline{\accpar})$ recognizes $\Sigma^\omega \setminus L(\aut_0)$, where $\overline{\accpar} = Q_0 \setminus \accpar$ and where $\overline{\delta_0}$ is the dual of $\delta_0$, i.e., $\overline{\delta_0}(q, A)$ is obtained from $\delta_0(q, A)$ by replacing each disjunction by a conjunction, each conjunction by a disjunction, each $\ttrue$ by $\ffalse$, and each $\ffalse$ by $\ttrue$. 
		
	\item The disjoint union of $\aut_0$ and $\aut_1$ with a fresh initial state~$q_\init$ and $\delta(q_\init, A) = \delta_0(q_\init^0, A) \wedge  \delta_1(q_\init^1, A)$ recognizes $L(\aut_0) \cap L(\aut_1)$.

	\item The disjoint union of $\aut_0$ and $\aut_1$ with a fresh initial state~$q_\init$ and $\delta(q_\init, A) = \delta_0(q_\init^0, A) \vee  \delta_1(q_\init^1, A)$ recognizes $L(\aut_0) \cup L(\aut_1)$.

\end{itemize}
The latter two constructions can obviously be generalized to unions and intersections of arbitrary arity while still only requiring a single fresh state. 

We prove the following lemma, which implies Theorem~\ref{theorem-translation-oldcor}, as alternating Büchi automata can be translated into non-deterministic Büchi automata of size $2^{\bigo(n)}$~\cite{MH}.

\begin{lemma}
\label{lemma-translation-alternating}
For every $\rldl$ formula~$\varphi$ and every $\tval \in \bool_4$, there is a weak alternating Büchi automaton~$\aut_{\varphi, \tval}$ with $\bigo(\size{\varphi})$ states recognizing the language~$\set{w \in (\pow{P})^\omega \mid \rldleval(w,\varphi) \succeq \tval} $.
\end{lemma}

\begin{proof}
We first construct the desired automaton by induction over the structure of $\varphi$. Then, we estimate its size. We begin by noting that $\aut_{\varphi,0000}$ is trivial for every formula~$\varphi$, as it has to accept every input. Hence, we only consider $\tval \succ 0000$ in the remainder of the proof.

For an atomic proposition~$p \in P$, $\aut_{p,\tval}$ for $\tval \succ 0000$ is an automaton that accepts exactly those $w$ with $p \in w(0)$. Such an automaton can easily be constructed.

Now, consider a negation~$\varphi = \neg \varphi'$: by definition, we have $\rldleval(w,\varphi) = 0000$ if $\rldleval(w, \varphi') = 1111$, and $\rldleval(w, \varphi) = 1111$ if $\rldleval(w, \varphi') \neq 1111$. Thus, $\aut_{\varphi, \tval}$ for $\tval \succ 0000$ has to accept the language~$\set{w \mid \rldleval(w, \varphi') \neq 1111}$, which is the complement of the language of $\aut_{\varphi',1111}$.
Hence, we obtain the desired automaton by applying the closure properties.

Next, let us consider a conjunction of the form $\varphi = \varphi_0 \wedge \varphi_1$. To this end, recall that $\rldleval(w, \varphi) = \min\set{\rldleval(w, \varphi_0) , \rldleval(w, \varphi_1)}$. Hence, $\aut_{\varphi,\tval}$ has to recognize the language~$L(\aut_{\varphi_0,\tval}) \cap L(\aut_{\varphi_1,\tval}) $.
Again, we obtain the desired automaton by applying the closure properties.

The construction for a disjunction~$\varphi = \varphi_0 \vee \varphi_1$ is dual to the conjunction: we have $\rldleval(w, \varphi) = \max\set{\rldleval(w, \varphi_0) , \rldleval(w, \varphi_1)}$ and thus construct $\aut_{\varphi,\tval}$ such that it recognizes the language $L(\aut_{\varphi_0,\tval}) \cap L(\aut_{\varphi_1,\tval}) $.
Hence, we obtain the desired automaton by applying the closure properties.

For an implication~$\varphi = \varphi_0 \implies \varphi_1$, we again implement the semantics via Boolean combinations of automata. Recall that $\rldleval(w, \varphi_0 \implies \varphi_1)$ is equal to 
 	$1111$ if $\rldleval(w,\varphi_0) \preceq \rldleval(w, \varphi_1)$. Otherwise, it is equal to 
 	$\rldleval(w, \varphi_1)$. 
 	
 	Here, we need to construct auxiliary automata~$\aut^=_{\varphi_i,\tval}$ that accept the traces~$w$ with $\rldleval(w,\varphi_i)=\tval$. 
 	For $\tval =1111$, this automaton is equal to $\aut_{\varphi_i, \tval}$ and for $\tval \prec 1111$ it is obtained by constructing the automaton recognizing $L(\aut_{\varphi_i,\tval}) \setminus \aut_{\varphi_i,\tval'}$, where $\tval'$ is the next-larger truth value after $\tval$. Here, the set difference is implemented by taking the intersection with the complement. 
Note that, unlike $\aut_{\varphi_i,\tval}$, which accepts $w$ if $\rldleval(w,\varphi_i)$ is greater or equal than $\tval$, the automaton~$\aut^=_{\varphi_i,\tval}$ only accepts~$w$ if $\rldleval(w,\varphi_i)$ is equal to $\tval$.
 	
 	Now, we combine the auxiliary automata to construct $\aut_{\varphi,\tval}$ so that it recognizes the language
 	\[\left(\bigcup_{\substack{\tval_0,\tval_1 \in \bool_4\\ \tval_0 \preceq \tval_1}} L(\aut^=_{\varphi_0,\tval_0}) \cap L(\aut^=_{\varphi_1,\tval_1})
 	\right) \cup 
 	\left(\bigcup_{\substack{\tval_0,\tval_1 \in \bool_4\\ \tval_0 \succ \tval_1 \preceq \tval}} L(\aut^=_{\varphi_0,\tval_0}) \cap L(\aut^=_{\varphi_1,\tval_1})\right).\]
 	The left part covers all cases in which the implication evaluates to $1111$. Due to $1111 \succeq \tval$ for every $\tval$, this part is equal for all automata. The right part covers all other cases, which depend on $\tval$.
 	So, we obtain the desired automaton by applying the closure properties.

Now, we turn to the constructions for the guarded temporal operators, which are more involved as we have to combine automata for guards, for the tests occurring in them, and for formulas. We follow the general construction presented by Faymonville and Zimmermann~\cite{FaymonvilleZimmermann17}, but generalize it to deal with the richer truth values underlying the robust semantics. 
Intuitively, for a diamond formula~$\ddiamonddot{r}\varphi'$ we construct an automaton that checks whether its input has a prefix that matches~$r$ (with the required degree) such that the corresponding suffix satisfies $\varphi'$ (again with the required degree). While checking for the prefix matching $r$, the automaton we construct also has to check that the tests in $r$ are satisfied (with the required degree). To this end, we use alternation to spawn copies of the automata we have already constructed for the tests. Finally, the construction for a box formula will be dual, but technically slightly more involved due to the robust semantics of the box-operator and the case distinctions involved in its definition.

First, we consider a diamond formula~$\varphi = \ddiamonddot{r}\varphi'$ with tests~$\theta_1, \ldots, \theta_n$ in $r$. Recall that we have $\rldleval(w, \varphi) = b_1 b_2 b_3 b_4$ where $b_i = \max_{j \in \robRexp_i(w,r)} \rldleval_i(\suff{w}{j}, \varphi')$ for all~$i \in \set{1,2,3,4}$. Thus, $\aut_{\varphi, \tval}$ has to accept $w$ if and only if $w$ has an $r$-match of degree~$\tval$ that is $\varphi'$-satisfying of degree~$\tval$.

By induction hypothesis, we have automata~$\aut_{\varphi',\tval}$ and $\aut_{\theta_j, \tval}$ for every test~$\theta_j$ in $r$. Also, we have an $\epsilon$-NFA with tests~$\autr_r$ equivalent to~$r$ due to Lemma~\ref{lemma-guards2automata}. We combine these automata to the alternating automaton~$\aut_{\varphi,\tval}$ by non-deterministically guessing a (finite) run of $\autr_r$. Whenever the run encounters a final state, the automaton may jump to the initial state of $\aut_{\varphi',\tval}$ and then behave like that automaton. Furthermore, while simulating $\autr_r$, $\aut_{\varphi,\tval}$ also has to verify that the tests occurring along the guessed run of~$\autr_r$ hold true by universally spawning copies of $\aut_{\theta_j,\tval}$ each time a state labeled with $\theta_j$ is traversed. 
Since we do not allow for $\epsilon$-transitions in alternating automata, we have to eliminate the $\epsilon$-transitions of $\autr_r$ during the construction of~$\aut_{\varphi, \tval}$. Finally, in order to prevent~$\aut_{\varphi, \tval}$ from simulating~$\autr_r$ ad infinitum, the states copied from $\autr_r$ are all rejecting, which forces the jump to $\aut_{\varphi, \tval}$ to be executed eventually.

Formally, we define $\aut_{\varphi,\tval} = (Q, \pow{P}, q_\init, \delta, \accpar)$ where
\begin{itemize}
	\item $Q$ is the disjoint union of the sets of states of the automata~$\autr_r$, $\aut_{\theta_j, \tval}$ for $j \in \set{1,\ldots,n}$, and $\aut_{\varphi', \tval}$,
	\item $q_\init$ is the initial state of $\autr_r$,
	\item $\accpar$ is the union of the accepting states of $\aut_{\theta_j, \tval}$ for $j \in \set{1,\ldots,n}$, and of $\aut_{\varphi', \tval}$,
\end{itemize}
and where $\delta$ is defined as follows: if $q$ is a state of $\autr_r$, then\label{elim}
\[\delta(q, A) = \begin{cases}
	\bigvee_{q' \in Q^r}\bigvee_{\pi \in \Pi(q,q')} \bigvee_{p \in \delta^r (q', A)} (p \wedge \bigwedge_{\theta_j \in \marking(\pi)} \delta^j(q_\init^j, A))& \\
	 \hspace{4.cm}\vee &\\
	\bigvee_{q' \in F^r}\bigvee_{\pi \in \Pi(q,q')} (\delta'(q_\init', A) \wedge \bigwedge_{\theta_j \in \marking(\pi)} \delta^j(q_\init^j, A))  & \\
	\end{cases}\]
	where $q_\init^j$ and $q_\init'$ are the initial states of $\aut_{\theta_j,\tval}$ and $\aut_{\varphi',\tval}$, respectively, where $Q^r$ ($F^r$) is the set of (final) states of $\autr_r$, where $\delta_r$, $\delta'$, and $\delta_j$ are the transition functions of $\autr_r$, $\aut_{\varphi',\tval}$, and $\aut_{\theta_j,\tval}$ respectively, and where the sets $\Pi(q,q')$ of $\epsilon$-paths are induced by $\autr_r$. Furthermore, for states~$q$ of $\aut_{\varphi',\tval}$, we define $\delta(q, A) = 	\delta'(q, A) $ and for states~$q$ of $\aut_{\theta_j,\tval}$ we define $\delta(q, A) = \delta^j(q, A)$. The resulting automaton accepts $w$ if and only if $w$ has at least one  $r$-match of degree~$\tval$ that is $\varphi'$-satisfying of degree~$\tval$ (cf.~\cite{FaymonvilleZimmermann17}, where the correctness is proven for plain $\ldl$).

Finally, we consider the box operator, which requires the most involved construction due to the case distinction defining the $b_i'$ and the subsequent maximization to obtain the~$b_i$. First, recall that the semantics of the box operator is not dual to the semantics of the diamond operator. Nevertheless, the dual construction of the one for the diamond operator is useful as a building block. We first present this construction before tackling the construction for the box operator.

In the dual construction, one interprets $\autr_r$ as a universal automaton whose transitions are ignored if the test on the source of the transition fails. Furthermore, each visit to a final state spawns a copy of the automaton~$\aut_{\varphi',\tval}$, as every $r$-match has to be $\varphi'$-satisfying. Thus, the states of $\autr_r$ are now accepting, as all $r$-matches have to be considered, and the automata for the tests are dualized in order to check for the failure of the test.

Formally, this approach yields the weak alternating Büchi automaton~$(Q, \pow{P}, q_\init, \delta, \accpar)$ where $Q$ and $q_\init$ are as above, where
\[\delta(q, A) = \begin{cases}
	\bigwedge_{q' \in Q^r}\bigwedge_{\pi \in \Pi(q,q')} \bigwedge_{p \in \delta^r (q', A)} (p \vee \bigvee_{\theta_j \in \marking(\pi)} \overline{\delta^j}(q_\init^j, A))& \\
	 \hspace{4.cm}\wedge & \\
	\bigwedge_{q' \in F^r}\bigwedge_{\pi \in \Pi(q,q')} (\delta'(q_\init', A) \vee \bigvee_{\theta_j \in \marking(\pi)} \overline{\delta^j}(q_\init^j, A))  & \\
	\end{cases}\]
for states~$q$ of $\autr_r$, where $q_\init^j$ and $q_\init'$ are the initial states of $\aut_{\theta_j,\tval}$ and $\aut_{\varphi',\tval}$, respectively.
Here, we use the fact that the final states of $\autr_r$ have no outgoing transitions, which implies that no match is missed by contracting an $\epsilon$-path.
Additionally, we define $\delta(q, A) = 	\delta'(q, A) $ for states~$q$ of $\aut_{\varphi',\tval}$, and $\delta(q, A) = \overline{\delta^j}(q, A)$ for states~$q$ of $\aut_{\theta_j,\tval}$.
Finally, the set of accepting states is the union of the set of all states of $\autr_r$, the rejecting states of the $\aut_{\theta_j,\tval}$, and the accepting states of $\aut_{\varphi',\tval}$. Recall that dualizing the transition relation and swapping accepting and rejecting states of the automata~$\aut_{\theta_j,\tval}$ amounts to complementation. This allows terminating runs of $\autr_r$ if a test does not hold true. The resulting automaton accepts a trace if and only if every $r$-match of degree~$\tval$ is $\varphi$-satisfying of degree~$\tval$ (cf.~\cite{FaymonvilleZimmermann17}, where the correctness is proven for plain $\ldl$)).

Now, we fix~$\varphi = \bboxdot{r}\varphi'$. Recall that we have $\rldleval(w, \varphi) = b_1 b_2 b_3 b_4$ with $b_i = \max\set{b_1', \ldots, b_i'}$ for some bits~$b_i'$. The maximization is easily implemented using the Boolean closure properties of alternating automata provided we have automata checking that some bit~$b_i'$ is equal to one. Two cases are trivial: Indeed, we have $b_1' = 1$ if and only if every $r$-match of degree~$1111$ is $\varphi$-satisfying of degree~$1111$. This property is checked by the dual automaton constructed above.
Furthermore, $b_4' = 1$ if and only if $ \rldleval(w, \ddiamonddot{r}\varphi') \succeq 0001$ or if there is no $r$-match of degree~$0001$. The former language is recognized by $\aut_{\ddiamonddot{r}\varphi',0001}$, the latter one by an automaton we construct below. We then combine these two automata to obtain $\aut_{\varphi,0001}$.

Hence, it remains to consider $b_2'$ and $b_3'$, which are both defined by a case distinction over the number of $r$-matches of the trace. These case distinctions are implemented using alternation. To this end, we first show how to test for the three cases, i.e., we argue that the following languages are recognizable by weak alternating Büchi automata, where $i \in \set{1,2,3,4}$:
\begin{enumerate}
	\item $L_i^\emptyset(r) = \set{w \in (\pow{P})^\omega \mid \size{\robRexp_i(w,r)} = 0}$.
	\item $L_i^f(r) = \set{w \in (\pow{P})^\omega \mid 0 < \size{\robRexp_i(w,r) } <  \infty}$.
	\item $L_i^\infty(r) = \set{w \in (\pow{P})^\omega \mid \size{\robRexp_i(w,r)} = \infty}$.
\end{enumerate}
Let $\theta_1, \ldots, \theta_n$ be the tests in $r$. By induction hypothesis, we have weak alternating Büchi automata~$\aut_{\theta_j, \tval}$ for every $\theta_j$ and every truth value~$\tval$. 

The first case is already solved, as we have, for each $i \in \set{1,2,3,4}$, $\robRexp_i(w,r) = \emptyset$ if and only if $\rldleval_i(w,\ddiamonddot{r}\ttrue ) = 0$, which is in turn equivalent to $\rldleval(w, \ddiamonddot{r}\ttrue) \prec   \itotruthvalue{i} $, i.e., the complement of the automaton $\aut_{\ddiamonddot{r}\ttrue, \itotruthvalue{i} }$ recognizes $L_i^\emptyset(r)$.

Next, we construct an automaton for the language~$L_i^\infty(r)$. Then, the automaton for~$L_i^f(r)$ is obtained as the intersection of the complement automata for the other two languages (for the given $r$ and $i$). Thus, we need to construct an automaton that accepts~$w$ if there are infinitely many $r$-matches of degree~$\itotruthvalue{i}$.

The construction of an automaton for $L_i^\infty(r)$ is more involved than the previous one, as the automaton~$\autr_r$ checking for matches with $r$ is non-deterministic. Nevertheless, we show that standard arguments about non-deterministic automata still yield the desired result. 
Intuitively, the automaton recognizing $L_i^\infty(r)$ has to determine whether infinitely many prefixes of $w$ are accepted by $\autr_r$ (while dealing with tests appropriately). This is implemented as follows: we start with $\autr_r$, eliminate $\epsilon$-transitions as in the case of diamond formulas on Page~\pageref{elim} (i.e., resulting in a non-deterministic choice ranging over all $\epsilon$-paths, each universally spawning a copy of $\aut_{\theta_j,  \itotruthvalue{i} }$ for each test~$\theta_j$ encountered along the corresponding $\epsilon$-path). Furthermore, on each transition, the run branches universally into a disjoint copy of this structure where all states are rejecting. Also, all states that have a transition leading to a state that was accepting in $\autr_r$ are equipped with a new transition leading to a fresh accepting sink state (with the same transition label and the same test automata being spawned).
Finally, the states of the original copy are all accepting. Thus,  the resulting automaton is weak.

Using König's Lemma, one can show that the resulting automaton accepts $w$ if and only if $w \in L_i^\infty(r)$. We leave the details to the industrious reader and just note that we have now constructed all automata we need to capture the cases in the case distinction.

Extending the construction just presented also allows us to construct an automaton that accepts a trace~$w$ if and only if it has infinitely many $\varphi'$-satisfying $r$-matches (both of degree~$\itotruthvalue{i}$). To this end, the copies spawned to check for matches are not equipped with transitions leading to an accepting sink, but with transitions leading to the initial state of $\aut_{\varphi',  \itotruthvalue{i} }$ to check for satisfaction of $\varphi'$. 
Similarly, we can construct an automaton that accepts a trace~$w$ if and only if it has infinitely many $r$-matches of degree~$\itotruthvalue{i}$ that are \emph{not} $\varphi'$-satisfying of degree~$\itotruthvalue{i}$.
Again, we leave the details to the reader.

These automata also allow us to construct an automaton that accepts a trace~$w$ if and only if $\robRexp_i(w,r)$ is infinite and almost all $r$-matches in $\robRexp_i(w,r)$ are $\varphi'$-satisfying (both of degree~$\itotruthvalue{i}$). This automaton is obtained by taking the automaton checking for infinitely many $\varphi'$-satisfying $r$-matches (both of degree~$\tval$) and intersecting it with the complement of the one checking for infinitely many $r$-matches that are not $\varphi'$-satisfying of degree~$\itotruthvalue{i}$.

Combining the automata checking the cases of the case distinction with the automata checking for $\varphi'$-satisfiability yields the desired automata for $b_2'$ and $b_3'$: A case distinction is easily implemented using the Boolean closure properties and all necessary auxiliary automata have been constructed above.

It is straightforward to verify that each automaton we construct  is weak.
Hence, it remains to argue that $\aut_{\varphi, \tval}$ is of linear size in $\size{\varphi}$. To this end, we say that a weak alternating Büchi automaton~$(Q', \Sigma, q_\init', \delta', \accpar')$ is a subautomaton of $(Q, \Sigma, q_\init, \delta, \accpar)$ if $Q' \subseteq Q$, $\delta'(q, A) = \delta(q,A)$ for every $q \in Q'$ and every $A \in \Sigma$, and $\accpar' = Q' \cap \accpar $. 

Inspecting the construction above shows that an automaton~$\aut_{\varphi, \tval}$ is built from automata for immediate subformulas (w.r.t.\ all truth values if necessary), a test automaton (if applicable), and a constant number of fresh states. Furthermore, if formulas share subformulas, then the construction can also share these subautomata. Hence, we obtain the desired linear upper bound on the size of $\aut_{\varphi, \tval}$.
\end{proof}

It is not straightforward that the equivalent non-deterministic Büchi automata as in Theorem~\ref{theorem-translation-oldcor} can be constructed efficiently, as the definition of the alternating automaton involves $\epsilon$-paths of arbitrary length. However, these can be restricted to paths of bounded lengths, as for every $\epsilon$-path there is one that has the same tests, but no cycles between them. 
Then, as it is done for the similar construction for $\pldl$~\cite{FaymonvilleZimmermann17}, one can show that the Büchi automata can be constructed on-the-fly in polynomial space, relying on the breakpoint construction~\cite{MH} turning an alternating Büchi automaton into a non-deterministic one. This is sufficient for our applications later on.

Furthermore, as non-deterministic Büchi automata can be translated into deterministic parity automata (see, e.g., \cite{GraedelThomasWilke02} for definitions), we obtain the following corollary of Theorem~\ref{theorem-translation-oldcor}.

\begin{corollary}
\label{corollary:rldl2detparity}
	Let $\varphi$ be an $\rldl$ formula, $n = \size{\varphi}$, and $\tval \in \bool_4$.
 There is a deterministic parity automaton~$\autp_{\varphi, \tval}$ with $2^{2^{\bigo(n)}}$ states and with $2^{\bigo(n )}$ colors recognizing the language~$\set{w \in (\pow{P})^\omega \mid \rldleval(w,\varphi) \succeq \tval}$.
\end{corollary}

The translations from logic to automata just proven allow us to study the expressiveness of $\rldl$ and solve its model checking and synthesis problem.

%% file: content/rldl-expressiveness.tex
In this section, we compare the expressiveness of $\rldl$ to that of $\rltl$ and $\ldl$. Following Tabuada and Neider~\cite{TabuadaNeider16} we focus on the fragment~$\rltl$ without next, until and release operators. While the next and until operator could be added easily, the robust semantics of the release operator is incompatible with our definition of the robust box operator. 
 It turns out as expected, that $\rldl$ subsumes $\rltl$ and $\rldl$. Conversely, every $\rldl$ formula~$\varphi$ can be translated into four $\ldl$ formulas~$\varphi_1, \ldots, \varphi_4$ that encode $\varphi$ in the following sense: We have $\rldleval_i(w,\varphi) = \ldleval(w, \varphi_i)$ for every $w$. 

\begin{theorem}
\label{thm-fragments}
Both $\rltl$ and $\ldl$ can be embedded into $\rldl$. 
\end{theorem}

\begin{proof}
Let us first embed $\rltl$ into $\rldl$ by showing that the syntactic embedding of $\ltl$ into $\ldl$ extends to robust semantics. Recall that $\rltl$ only has temporal operators~$\Diamonddot$ and $\Boxdot$, which we replace by $\ddiamonddot{\ttrue^*}$ and $\bboxdot{\ttrue^*}$. Note that $\robRexp_i(w,\ttrue^*) = \nats$ holds true for every $w$ and every $i$. Hence, a straightforward induction shows that the resulting $\rldl$ formula is equivalent to the original $\rltl$ formula. In particular, only the first case in the case distinctions defining the semantics of the robust box operator of $\rldl$ is used, which mimics the definition of the semantics of the always operator in $\rltl$.

Conversely, using a straightforward induction over the structure of $\ldl$ formulas, we can show that $\ldleval(w, \varphi) = \rldleval_1(w, \varphi')$ for every~$w$ and every $\ldl$ formula~$\varphi$, where $\varphi'$ is the $\rldl$ formula obtained from $\varphi$ by replacing each $\ddiamond{r}$ with $\ddiamonddot{r}$, each $\bbox{r}$ with $\bboxdot{r}$, and each implication $\psi_1 \to \psi_2$ with $\lnot \psi_1 \lor \psi_2$.
This shows that $\ldl$ can be embedded into $\rldl$.
Note, however, that we need to replace each implication with a negation and a disjunction.
This is necessary to account for the more complex definition of implications in $\rltl$/$\rldl$.
\end{proof}

As $\ltl$ is a semantic fragment of $\ldl$, we immediately obtain that $\ltl$ can be embedded into $\rldl$ and, thus, $\rldl$ inherits the lower bounds of $\ltl$. 

Our next theorem states that $\ldl$ and $\rldl$ are of equal expressiveness. 
The direction from $\ldl$ to $\rldl$ was shown in Theorem~\ref{thm-fragments}, hence we focus on the other one.
Following Tabuada and Neider~\cite{TabuadaNeider16}, we construct for every $\rldl$ formula~$\varphi$ four $\ldl$ formulas~$\varphi_1, \ldots, \varphi_4$ encoding $\varphi$ as explained above. The construction relies on Theorem~\ref{theorem-translation-oldcor}, unlike the analogous result translating robust $\ltl$ directly into $\ltl$~\cite{TabuadaNeider16}.

\begin{theorem}\label{thm:rLDL-LDL-equally-expressive}
$\ldl$ and $\rldl$ are equally expressive and the translations are effective.
\end{theorem}

\begin{proof}
As argued above, we only have to consider the direction from $\rldl$ to $\ldl$. Hence, fix an $\rldl$ formula~$\varphi$ and $i \in \set{1,2,3,4}$. Due to Theorem~\ref{theorem-translation-oldcor}, $\set{w \in (\pow{P})^\omega \mid \rldleval_i(w, \varphi) = 1}$ is $\omega$-regular. Hence, due to $\ldl$ being equi-expressive to the $\omega$-regular languages~\cite{Vardi11}, there is also an $\ldl$ formula~$\varphi_i$ with $\ldleval(w, \varphi_i) = 1$ if and only if $\rldleval_i(w, \varphi) = 1$. Hence, $\varphi_i$ has the desired properties.
\end{proof}

Let us analyze the complexity of the translation in more detail.
A non-deterministic Büchi automaton for an $\rldl$ formula is in general of exponential size and has to be determinized before it can be translated into $\ldl$ (say with max-parity acceptance), which incurs a second exponential blowup.
The resulting deterministic automaton can then be translated into $\ldl$ with an exponential blowup.
The resulting formula expresses that the unique run on the input satisfies the following properties: there is an even color~$c$ and a position~$j$ such that after $j$ no larger color appears, $c$ appears at least once, and every time color~$c$ appears, it is not the last occurrence of $c$ (see~\cite{WeinertZimmermann16} for a similar construction).
All three properties are easily expressed in $\ldl$ by constructing guards~$r_{q,q'}$ that match infixes such that processing the infix starting in $q$ leads to the state~$q'$.
The number of subformulas is polynomial, but the guards~$r_{q,q'}$ have in general exponential size (both measured in the size of the doubly-exponential deterministic automaton).
Hence, the full construction incurs a triply-exponential blowup.
It is open whether this is unavoidable. 
On a more positive note, the resulting $\ldl$ formula is test-free, i.e., it does not contain tests in its guards.

We leave the question of whether there are non-trivial lower bounds on the translation for future work. For the special case of translating~$\rltl$ into $\ltl$ mentioned above, there is only a linear blowup. This translation was presented by Tabuada and Neider~\cite{TabuadaNeider16}, but they only claimed an exponential upper bound. However, closer inspection shows that it is linear if the size of formulas is measured in the number of distinct subformulas, not the length of the formula.

%% file: content/rldl-mcsynt.tex

Theorem~\ref{thm:rLDL-LDL-equally-expressive} immediately provides solutions for typical applications of $\rldl$, such as model checking and synthesis, by reducing the problem from the domain of $\rldl$ to that of $\ldl$. However, the price to pay for this approach is a triply-exponential blow-up in the size of the resulting $\ldl$ formula, which is clearly prohibitive for any real-world application. For this reason, we now develop more efficient model checking and synthesis techniques that are based on our direct translation of $\rldl$ into automata (Theorem~\ref{theorem-translation-oldcor}).

We begin with the $\rldl$ model checking checking problem, which is defined as follows.

\begin{problem} \label{prob:rLDL-model-checking}
Let $\varphi$ be an $\rldl$ formula, $\sys$ a transition system, and let $\tval \in \bool_4$. Does $\rldleval(\trace(\rho), \varphi) \succeq \tval$ hold true for all paths~$\rho \in \Pi_\sys$?
\end{problem}

Using the translation of $\rldl$ formulas to weak alternating Büchi automata and subsequently to non-deterministic Büchi automata, Problem~\ref{prob:rLDL-model-checking} can be solved as follows:
\begin{enumerate}
	\item \label{item:rLDL-model-checking:1} Translate the transition system $\sys$ into a non-deterministic Büchi automaton $\autb_\sys$ with $L(\autb_\sys) = \{ \trace(\rho) \in (\pow{P})^\omega \mid \rho \in \Pi_\sys \}$ in the usual way: $\autb_\sys$ has the same states as $\sys$, the transitions are $\{ (s, \lambda(s), s') \mid  (s, s') \in E \}$, and all states are accepting.
	\item \label{item:rLDL-model-checking:2} Construct the weak alternating Büchi automaton $\aut_{\varphi, \tval}$ accepting the language $\{ w \in (2^P)^\omega \mid \rldleval(w, \varphi) \succeq \tval \}$.
	\item \label{item:rLDL-model-checking:3} Complement $\aut_{\varphi, \tval}$ to obtain a weak alternating Büchi  automaton~$\overline{\aut_{\varphi, \tval}}$ accepting the language $\{ w \in (2^P)^\omega \mid \rldleval(w, \varphi) \prec \tval \}$.
	\item \label{item:rLDL-model-checking:4} Convert $\overline{\aut_{\varphi, \tval}}$ into an equivalent non-deterministic Büchi automaton $\overline{\autb_{\varphi, \tval}}$ and compute the product automaton $\autb$ with $L(\autb) = L(\autb_\sys) \cap L(\overline{\autb_{\varphi, \tval}})$ in the usual way.
	\item \label{item:rLDL-model-checking:5} Check whether $L(\autb) = \emptyset$ using a standard algorithm such as a nested depth-first search~\cite{BaierKatoen08}. The answer to Problem~\ref{prob:rLDL-model-checking} is ``yes'' if and only if $L(\autb) = \emptyset$.
\end{enumerate}

The number of states of the weak alternating Büchi automata in Step~\ref{item:rLDL-model-checking:2} and \ref{item:rLDL-model-checking:3} is both in $\bigo{(\size{\varphi})}$.
Thus, the number of states of the non-deterministic Büchi automaton $\overline{\autb_{\varphi, \tval}}$ constructed in Step~\ref{item:rLDL-model-checking:4} is in $2^{\bigo(\size{\varphi} )}$, and that of $\autb$ is in $\size{\sys} \cdot 2^{\bigo{(\size{\varphi} )}}$, where $\size{\sys}$ denotes the number of states of the transition system $\sys$ (cf.\ Theorem~\ref{theorem-translation-oldcor}).
Finally, the time required for the emptiness check in Step~\ref{item:rLDL-model-checking:5} is quadratic in the number of states of $\autb$ (linear in the number of $\autb$'s transitions).
Consequently, the $\rldl$ model checking problem can be solved in time $\size{\sys}^2 \cdot 2^{\bigo{(\size{\varphi} )}}$ and, hence, is in $\exptime$.

We now show that the problem is not only in $\exptime$, but that it is, in fact, $\pspace$-complete.
To this end, we leverage the exponential compilation property (see Theorem~\ref{theorem-translation-oldcor}) and standard on-the-fly techniques for checking emptiness of exponentially-sized Büchi automata~\cite{VardiWolper94}, which yield a $\pspace$ upper bound on the complexity of Problem~\ref{prob:rLDL-model-checking}.
The matching lower bound follows from the subsumption of $\ldl$ shown above, as model checking $\ldl$ is $\pspace$-complete.

\begin{theorem} \label{theorem:rldl-model-checking-complexity}
$\rldl$ model checking is $\pspace$-complete.
\end{theorem}

\begin{proof}
As shown above, the $\rldl$ model checking problem is in $\exptime$.
To show membership in $\pspace$, we use the observation that given two states of $\autb$, one can decide in polynomial space whether the second state is a successor of the first one (cf. Vardi and Wolper~\cite{VardiWolper94}).
Moreover, one can represent states of $\autb$ in polynomial space. This allows running the classical model checking algorithm, which searches for a counterexample, in polynomial space by guessing an appropriate run.

$\pspace$-hardness, on the other hand, follows immediately from the facts that (a) the $\ltl$ semantics is embedded in $\rldl$, via the embedding of $\ldl$ in $\rldl$ (see Theorem~\ref{thm:rLDL-LDL-equally-expressive}), and (b) $\ltl$ model checking is $\pspace$-hard~\cite{SistlaClarke85}.
Thus, $\rldl$ model checking is $\pspace$-hard as well.
\end{proof}

Before we move on to reactive synthesis, let us briefly remark that the model checking problem for $\rltl$ is defined slightly differently. Instead of asking whether $\rltleval(\lambda(\rho), \varphi) \succeq \tval$, Tabuada and Neider~\cite{TabuadaNeider16} fix a set $B \subseteq \bool_4$ and ask whether $\rltleval(\lambda(\rho), \varphi) \in B$ for all paths $\rho \in \Pi_\sys$.
However, this slightly more general problem can easily be answered by a simple adaptation of Step~\ref{item:rLDL-model-checking:2} of the procedure above: given a (finite) set $B \subseteq \bool_4$, we construct a weak alternating Büchi  automaton accepting the language $\{ w \in (2^P)^\omega \mid \rldleval(w, \varphi) \in B \}$ using Boolean combinations of the automata $\aut_{\varphi, \tval}$. Then, it is not hard to verify that this variant of the $\rldl$ model checking problem is also $\pspace$-complete.

Similar to model checking, the translation from $\rldl$ formulas to automata provides us with an effective means to synthesize reactive controllers from $\rldl$ specifications, i.e., for the following problem, where an $\rldl$ game has the form~$(\ggraph, \varphi, \beta)$ and Player~$0$ wins a play if and only if its trace~$w$ satisfies $\rldleval(w, \varphi) \ge \beta$.

\begin{problem} \label{prob:rLDL-synthesis}
Let $\game$ be an $\rldl$ game and $v$ a vertex.  Determine whether Player~$0$ has a winning strategy for $\game$ from $v$ and compute a finite-state winning strategy if so.
\end{problem}

Corollary~\ref{corollary:rldl2detparity} provides a straightforward way to solve Problem~\ref{prob:rLDL-synthesis} by reducing it to solving classical parity games (again, see~\cite[Chapter~2]{GraedelThomasWilke02} for an introduction to parity games) while the lower bound follows from the subsumption of $\ldl$.

\begin{theorem}\label{theorem-rldlgames}
Solving $\rldl$ games is $\twoexp$-complete.
\end{theorem}

\begin{proof}
We proceed in several steps constructing a reduction to a parity game.
\begin{enumerate}
	\item \label{item:rLDL-games:1} Construct the deterministic parity automaton $\autp_{\varphi, \tval}$ recognizing the language $\{ w \in (\pow{P})^\omega \mid \rldleval(w,\varphi) \succeq \tval \}$ according to Corollary~\ref{corollary:rldl2detparity}.
	\item \label{item:rLDL-games:2} Construct the product of $\autp_{\varphi, \tval} = (Q, \pow{P}, q_\init, \delta, \Omega)$ and the labeled game graph $\ggraph = (V_0, V_1, E, \lambda)$. This product is a classical (non-labeled) parity game $\game' = (\ggraph', \Omega')$ consisting of a game graph $\ggraph' = (V'_0, V'_1, E')$ with $V'_0 = V_0 \times Q$, $V'_1 = V_1 \times Q$, and $E' = \{ ((v, q), (v', \delta(q, \lambda(v)))) \mid (v, v') \in E \}$ as well as a parity winning condition $\Omega'$ with $\Omega'((v, q)) = \Omega(q)$ for each $(v, q) \in V'_0 \cup V'_1$.
	One obtains a play $\rho$ in the original game $\game$ from a play $\rho'$ in the extended game $\game'$ by projecting the vertices of $\rho'$ onto the first component.
	Thus, Player~0 wins a play $\rho'$ in~$\mathcal G'$ from a vertex $(v, q_\init)$ if and only if the trace $\trace(\rho)$ obtained from the corresponding play~$\rho$ in $\mathcal G$ satisfies $\rldleval(\trace(\rho), \varphi) \succeq \tval$.
	\item \label{item:rLDL-games:3} Solve the game $\mathcal G'$ with standard algorithms for parity games, e.g., the recent quasi-polynomial time algorithm~\cite{DBLP:conf/stoc/CaludeJKL017}. Finally, check and return whether Player~$0$ has a winning strategy from vertex $(v, q_\init)$ and return a finite-state winning strategy if so.
\end{enumerate}

The above reduction is a standard game reduction, whose correctness can be shown using standard techniques. In fact, it provides a $\twoexp$ algorithm to solve Problem~\ref{prob:rLDL-synthesis}: due to Corollary~\ref{corollary:rldl2detparity}, the deterministic parity automaton $\autp_{\varphi, \tval}$ constructed in Step~\ref{item:rLDL-games:1} has $2^{2^{\bigo{(\size{\varphi} )}}}$ states and $2^{\bigo{(\size{\varphi} )}}$ colors; consequently, the parity game $\game'$ of Step~\ref{item:rLDL-games:2} has $\size{V} \cdot 2^{2^{\bigo{(\size{\varphi} )}}}$ vertices and $2^{\bigo{(\size{\varphi} )}}$ colors; thus, using the quasi-polynomial algorithm for solving parity games~\cite{DBLP:conf/stoc/CaludeJKL017} results in a doubly-exponential algorithm for Problem~\ref{prob:rLDL-synthesis}.

On the other hand, the fact that $\rldl$ subsumes $\ldl$ and, hence, $\ltl$ immediately implies that solving $\rldl$ games is $\twoexp$-hard since solving $\ltl$ games is already $\twoexp$-hard~\cite{PnueliRosner89a}.
\end{proof}

%% file: content/towardsrpromptldl.tex
In the previous sections, we studied robust $\ldl$, i.e., we combined robustness and increased expressiveness, and robust $\prompt$, i.e., we combined robustness and quantitative operators. The third combination of two aspects, i.e., quantitative operators and increased expressiveness, has been studied before~\cite{FaymonvilleZimmermann17}. For all three resulting logics, model checking and synthesis have the same complexity as for plain $\ltl$. 

Here, we consider the combination of all three extensions, obtaining the logic $\rpromptldl$, robust $\promptldl$.
As this logic is an extension of $\prompt$, it features negations only at the level of atomic propositions and does not allow implications.
The formulas of \rpromptldl\ are given by the grammar
\begin{align*}
\varphi &\cceq p \mid \neg p \mid \varphi \wedge \varphi \mid \varphi \vee \varphi \mid   \ddiamonddot{r} \varphi 
  \mid \bboxdot{r} \varphi \mid \promptddiamonddot{r} \varphi \\
    r & \cceq \phi \mid \varphi? \mid r+r \mid r \conc r \mid r^*
\end{align*} 
where $p$ again ranges over the atomic propositions in $P$ and $\phi$ over propositional formulas over $P$. Furthermore, the size of a formula is defined as for $\ldl$ and $\rldl$. 

The semantics are defined as expected: We add a bound~$k$ to the semantics of $\rldl$, which bounds the scope of the prompt diamond operator. 
\begin{itemize}
\item $\rpromptldleval(w,k,p) =\begin{cases}
	1111 &\text{if $p \in w(0)$,}\\
	0000 &\text{if $p \notin w(0)$,}
\end{cases}$ \quad\quad $\text{\normalfont\bfseries --} \,\,\, \rpromptldleval(w,k, \neg p) =\begin{cases}
	0000 &\text{if $p \in w(0)$,}\\
	1111 &\text{if $p \notin w(0)$,}
\end{cases}$

\item
$\rpromptldleval(w, k,\varphi_0 \wedge \varphi_1) = \min\set{\rpromptldleval(w, k,\varphi_0), \rpromptldleval(w,k,\varphi_1) }$, 
\item $\rpromptldleval(w,k, \varphi_0 \vee \varphi_1) = \max\set{\rpromptldleval(w,k, \varphi_0), \rpromptldleval(w,k,\varphi_1) }$, 

\item $\rpromptldleval(w,k,  \ddiamonddot{r}\varphi) = b_1 b_2 b_3 b_4$ where $b_i = \max\nolimits_{j \in \robpromptRexp_i(w,k,r)} \rpromptldleval_i(\suff{w}{j}, k,\varphi)$,

\item $\rpromptldleval(w,k, \bboxdot{r}\varphi) = b_1 b_2 b_3 b_4$ with $b_i = \max\set{b_1', \ldots, b_i'}$ for every $i \in \set{ 1,2,3,4}$, where
\begin{itemize}
	
	\item $b_1' = \min_{j \in \robpromptRexp_1(w,k,r)} \rpromptldleval_1(\suff{w}{j},k,\varphi)$,
	
	\item $b_2' = \begin{cases}
	\max_{j' \in \nats} \min_{j \in \robpromptRexp_2(w,k,r) \cap \set{j', j'+1, j'+2, \ldots}} \rpromptldleval_2(\suff{w}{j},k,\varphi)&\text{if $\size{\robpromptRexp_2(w,k,r)} = \infty$},\\
	
	\min_{j \in \robpromptRexp_2(w,k,r)} \rpromptldleval_2(\suff{w}{j},k,\varphi)&\text{if $0 < \size{\robpromptRexp_2(w,k,r)} < \infty$},\\
	
	1 &\text{if $\size{\robpromptRexp_2(w,k,r)} = 0$},
	\end{cases}$
	
	\item $b_3' = \begin{cases}
	\min_{j' \in \nats} \max_{j \in \robpromptRexp_3(w,k,r) \cap \set{j', j'+1, j'+2, \ldots}} \rpromptldleval_3(\suff{w}{j},k,\varphi)&\text{if $\size{\robpromptRexp_3(w,k,r)} = \infty$},\\
	
	\max_{j \in \robpromptRexp_3(w,k,r)} \rpromptldleval_3(\suff{w}{j},k,\varphi)&\text{if $0 < \size{\robpromptRexp_3(w,k,r)} < \infty$},\\
	
	1 &\text{if $\size{\robpromptRexp_3(w,k,r)} = 0$},
	\end{cases}$
	
	\item $b_4' =\begin{cases} \max_{j \in \robpromptRexp_4(w,k,r)} \rpromptldleval_4(\suff{w}{j},k,\varphi) &\text{if $\size{\robpromptRexp_4(w,k,r)} > 0$,}\\
1&\text{if $\size{\robpromptRexp_4(w,k,r)} = 0$,}
\end{cases}
$ and
\end{itemize}

\item $\rpromptldleval(w,k,  \promptddiamonddot{r}\varphi) = b_1 b_2 b_3 b_4$ where 
$b_i = \max\nolimits_{j \in \robpromptRexp_i(w,k,r) \cap \set{0, \ldots, k}} \rpromptldleval_i(\suff{w}{j}, k,\varphi)$.
\end{itemize}
Here, we adapt the definition of $\robRexp_i$ to account for the parameter~$k$: $\robpromptRexp_i(w,k,\varphi?) = \set{0}$ if $\rpromptldleval_i(w,k, \varphi)=1$ and $\robpromptRexp_i(w,k,\varphi?) =  \emptyset$ otherwise. All other cases are defined as before, but propagate the parameter~$k$.

$\rldl$ without negation and implication (and $\ldl$, for which negation and implication can be eliminated) and $\promptldl$ formulas can easily be translated into equivalent $\rpromptldl$ formulas.
$\prompt$, however cannot necessarily be translated into equivalent $\rpromptldl$ formulas, as the semantics of the release operator is not compatible with the semantics of $\rpromptldl$.

\begin{example}
\label{example-rpromptldl}
Consider the formula~$\bboxdot{((\neg t)^*\conc t\conc (\neg t)^*\conc t )^*} \promptddiamonddot{\ttrue^*} s $ and interpret $t$ as the tick of a clock and $s$ as a synchronization. Then, the formula intuitively expresses that every other tick of the clock is followed after a bounded number of steps (not ticks!) by a synchronization. 

More formally, the different degrees of satisfaction of $\varphi$ express the following possibilities, with respect to a given bound~$k$:
(i) every even clock tick is followed by a synchronization within $k$ steps;
(ii) almost every even clock tick is followed by a synchronization within $k$ steps;
(iii) infinitely many even clock ticks are followed by a synchronization within $k$ steps;
(iv) there is at least one even clock tick that is followed by a synchronization within $k$ steps.

This property can neither be expressed in (robust) $\ldl$ nor in (robust) $\prompt$. Also note that unlike for the similar formula from Example~\ref{example-rprompt}, the last two possibilities are not trivial, as we now only consider positions with an even clock tick and not all positions. 
\end{example}

In the previous sections, we have seen two approaches to translating robust logics into Büchi automata, the direct and the reduction-based one. Both are extensions of translations originally introduced by Tabuada and Neider for robust $\ltl$. The former one translates a formula of a robust logic directly into an equivalent Büchi automaton while the latter one first translates a formula of a robust logic  into an equivalent classical (non-robust) logic, for which a translation into equivalent Büchi automata is already known. For robust $\ltl$, both approaches are applicable~\cite{TabuadaNeider16} and yield Büchi automata of exponential size. 
Here, out of necessity, we apply both approaches: for robust $\ldl$, we present a direct translation while we present a reduction-based approach for robust $\prompt$. Let us quickly elaborate the reasons for this.

First, consider the reduction-based approach for robust $\ltl$, which translates a formula~$\varphi$ of robust $\ltl$ and a truth value~$\beta \succ 0000$ into an $\ltl$ formula~$\varphi_\beta$ that captures $\varphi$ with respect to $\beta$. To this end, the formula~$\varphi_\beta$ implements the intuitive meaning of the robust semantics for the always operator, e.g., we have $(\Boxdot p )_{1111} = \Box p$, $(\Boxdot p )_{0111} = \Diamond\Box p$, $(\Boxdot p )_{0011} = \Box\Diamond p$, and $(\Boxdot p )_{0001} = \Diamond p$. 

Trying to apply this approach to the  $\rldl$ formula~$\varphi = \bboxdot{r} p$, say for $\beta = 0111$, would imply using a formula of the form~$\ddiamonddot{r_0} \bboxdot{r_1} p$ where $r_0$ and $r_1$ are obtained by \myquot{splitting} up $r$. It captures the robust semantics of $\varphi$ with respect to $\beta$ on some trace~$w$ by expressing that there is an $r_0$-match~$j$ such that every $r_1$-match in $\suff{w}{j}$ is $p$-satisfying with degree~$\beta$. Thus, $r_0$ and $r_1$ have to be picked such that the $r_1$-matches in $\suff{w}{j}$ as above correspond exactly to the $r$-matches in $w$.
Further, to obtain a translation of optimal complexity, $r_0$ and $r_1$ have to be of polynomial size in $\size{r}$. It is an open problem whether such a splitting is always possible, in particular in the presence of tests in $r$ and guards with only finitely many $r$-matches.

Secondly, recall that the direct approach to robust $\ltl$ translates a formula~$\varphi$ of $\rltl$ into a Büchi automaton that captures $\varphi$ with respect to all $\beta \in \bool_4$ (by considering five initial states, one for each $\beta$). Trying to apply this approach to robust $\prompt$ requires using a more general automaton model that is able to capture the quantitative nature of the prompt diamond operator while still yielding a model checking and a synthesis algorithm with the desired complexity. To the best of our knowledge, no such translation from $\prompt$ to automata has been presented in the literature, which would be a special case of our construction here.

Thus, according to the state-of-the-art, the direct approach is the only viable one for robust extensions of $\ldl$ while the reduction-based approach is the only viable one for robust extensions of $\prompt$.
This leaves us with no viable approach for $\rpromptldl$. 

Nevertheless, we identify a fragment of $\rpromptldl$ for which both the model checking and the synthesis problem are decidable.
We obtain this fragment by disallowing tests in guards and by requiring them to always have infinitely many matches.
For such formulas, one can translate the guard into a deterministic finite automaton (without tests) and then use this automaton to \myquot{split} $r$.
However, this involves multiple exponential blowups and hence does not prove that the fragment has the exponential compilation property.
Nonetheless, this translation shows that both model checking and synthesis are decidable for this fragment.
The decidability of these problems for full $\rpromptldl$ is left for further research and seemingly requires new approaches.

\subsection{Restricting Guards in \boldmath$\rpromptldl$}
\label{subsec-fragment}

We say that a guard~$r$ is test-free if it does not contain tests as atoms, but only propositional formulas over the atomic propositions. A formula is test-free if each of its guards is test-free. In the remainder, we only consider test-free formulas. As the adaptions made to define $\robpromptRexp_i$ are only concerned with tests, they can be ignored when reasoning about test-free formulas.

\begin{remark}
Let $r$ be a test-free guard. Then, $\robpromptRexp_i(w,k,r)$ is independent of~$i$ and~$k$ for every trace~$w$. 	
\end{remark}

Hence, in the following, we use $\Rexp(w,r)$ (as defined for $\ldl$) instead of $\robpromptRexp_i(w,k,r)$, since the definitions coincide for test-free guards. 

We say that a test-free guard~$r$ is limit-matching if we have $\size{\Rexp(w,r)} = \infty$ for every trace~$w$.
This is well-defined due to the previous remark.
Again, a test-free formula is limit-matching if each of its guards is limit-matching. 

\begin{lemma}
\label{lemma-syntaxeffective}
The problem \myquot{Given a test-free formula~$\varphi$, is $\varphi$ limit-matching?} is in~$\pspace$.
\end{lemma}

\begin{proof}
The problem is in $\pspace$ if one can decide in polynomial space whether a single test-free guard is limit-matching. Hence, let $r$ be such a guard, which is limit-matching if and only if infinitely many prefixes of each trace~$w$ match $r$. An application of König's Lemma yields that the latter condition is equivalent to each $w$ being Büchi-accepted by $\autr_r$. Due to test-freeness, $\autr_r$ can indeed be seen as a Büchi automaton with $\epsilon$-transitions. Hence, $r$ is limit-matching if and only if $\autr_r$ is universal, which can be decided in polynomial space~\cite{SistlaVardiWolper85} (after eliminating $\epsilon$-transitions). The automaton being of the same size as the guard and being efficiently constructible concludes the proof.
\end{proof}

\begin{example}
Recall the formula~$\varphi = \bboxdot{((\neg t)^*\conc t\conc (\neg t)^*\conc t )^*} \promptddiamonddot{\ttrue^*} s$ from Example~\ref{example-rpromptldl}. It is test-free, but not limit-matching as traces with finitely many $t$ only have finitely many $((\neg t)^*\conc t\conc (\neg t)^*\conc t )^*$-matches.

Nevertheless, test-free and limit-matching $\rpromptldl$ formulas can make use of arbitrary modulo counting, a significant advance in expressiveness over classical $\ltl$, thus witnessing the usefulness of the fragment.

For example, the formula~$\bboxdot{r}\promptddiamonddot{r}s$ with $r = (\ttrue\conc\ttrue)^*$ expresses, when evaluated with respect to a bound~$k$, that the distance between synchronizations at even positions is bounded by $k$, i.e., we use the test-free limit-matching guards to \myquot{filter out} the odd positions. 
\end{example}

Let us  note that for $\ldl$, the test-free fragment is of equal expressive power as full $\ldl$, albeit potentially less succinct. This claim follows easily from translating Büchi automata into $\ldl$ formulas, which results in test-free formulas. 

In the following, we consider the model checking and the synthesis problem for test-free limit-matching formulas. To this end, we proceed as in the case of $\rprompt$: We reduce these problems to those for $\promptldl$, i.e., we present a reduction-based translation to Büchi automata.

 Due to only considering limit-matching formulas, we do not have to deal with the cases of having only finitely many matches of a guard. On the other hand, we have to \myquot{split} guards to capture the semantics of the robust diamond operator (recall the discussion in Section~\ref{sec-towardsrpldl}). Here, we exploit the formula under consideration being test-free.

The main technical result on this fragment states that the logic can be derobustified, i.e., translated into $\promptldl$.

\begin{theorem}
\label{thm-prldl2pldl}
For every test-free limit-matching $\rpromptldl$ formula~$\varphi$ and every $\tval \in \bool_{4}$, there is a $\promptldl$ formula~$\varphi_\tval$ such that $\rpromptldleval(w, k, \varphi) \succeq \tval$ if and only if $\promptldleval(w, k, \varphi_\tval) = 1$.
\end{theorem}
	
\begin{proof}

Before we present the translation, we need to explain how to \myquot{split} guards, which is necessary to implement the semantics of the robust box operator (recall the discussion in Section~\ref{sec-towardsrpldl}).
For example, we have to check that almost all $r$-matches are $\psi$-satisfying for some guard~$r$ and some subformula~$\psi$. In $\ltl$, \myquot{almost all} is expressed by $\Diamond\Box$. We will use the analogous $\ldl$ operators, i.e., a formula of the form~$\ddiamond{ \cdot} \bbox{\cdot }$. But now we need guards~$r_0$ and $r_1$ for the diamond and the box operator so that the concatenation $r_0r_1$ is equivalent to $r$. To this end, we transform $r$ into a deterministic automaton. Then, for each state~$q$ of that automaton there exists a guard~$r_{q_\init,q}$ capturing the words leading from the initial state to $q$, and a guard~$r_{q,F}$ capturing all words leading from $q$ to an accepting state.  Ultimately, we end up with a formula of the form~$\bigvee_q\ddiamond{ r_{q_\init,q}} \bbox{ r_{q,F} }$. 

Let $r$ be a test-free guard. Applying Lemma~\ref{lemma-guards2automata} to $r$ yields an $\epsilon$-NFA~$\autr_r$ without tests. Hence, eliminating $\epsilon$-transitions and determinizing the resulting automaton yields a deterministic finite automaton~$\autd_r$ such that $\pref{w}{j}$ is accepted by $\autd_r$ if and only if $j \in \Rexp(w,r)$. Furthermore, due to test-freeness, acceptance of $\pref{w}{j}$ by $\autd_r$ only depends on the prefix~$\pref{w}{j}$ of $w$, but not on the corresponding suffix~$\suff{w}{j}$. This property, which underlies the following construction, does not hold true for guards with tests. 

Now, let $Q$ be the set of states of $\autd_r$, $q_\init$ the initial state, and $F$ the set of final states. Then, for every $q \in Q$, one can efficiently construct regular expressions (i.e., guards)~$r_{q_\init,q}$ and $r_{q,F}$ such that $w \in (\pow{P})^*$ is in the language of $r_{q_\init,q}$ (of $r_{q,F}$) if and only if the unique run of $\autd_r$ starting in $q_\init$ (in $q$) ends in $q$ (in $F$).

Now, we are ready to construct~$\varphi_\tval$. Again, the case~$\tval = 0000$ is trivial. Hence, we assume $\tval \succ 0000$ in the following. We proceed by induction over the construction of the formula:
\begin{itemize}

	\item $p_\tval = p$ and $(\neg p)_\tval = \neg p$ for all atomic propositions~$p \in P$ and all $\tval \succ 0000$. 
	
	\item $(\varphi_0 \wedge \varphi_1)_\tval = (\varphi_0)_\tval \wedge (\varphi_1)_\tval$ for all $\tval \succ 0000$. 

	\item $(\varphi_0 \vee \varphi_1)_\tval = (\varphi_0)_\tval \vee (\varphi_1)_\tval$ for all $\tval \succ 0000$. 
	
	\item $(\ddiamonddot{r}\varphi )_\tval = \ddiamond{r}(\varphi_\tval)$ for all $\tval \succ 0000$
	
	 \item $(\bboxdot{r} \varphi)_{1111} = \bbox{r}(\varphi_{1111})$, 

	 \item 	$(\bboxdot{r} \varphi)_{0111} = \bigvee_{q \in Q} \ddiamond{r_{q_\init,q}}\bbox{r_{q,F}}(\varphi_{0111})$, where $\autd_r = (Q, \pow{P}, q_\init, \delta, F)$,

	 \item 	$(\bboxdot{r} \varphi)_{0011} = \bigwedge_{q \in Q} \bbox{r_{q_\init,q}}\ddiamond{r_{q,F}} (\varphi_{0011})$, where $\autd_r = (Q, \pow{P}, q_\init, \delta, F)$,

	 \item 	$(\bboxdot{r} \varphi)_{0001} = \ddiamond{r} (\varphi_{0001})$, and
	
	\item $(\promptddiamonddot{r} \varphi )_\tval = \promptddiamond{r} (\varphi_\tval)$ for all $\tval \succ 0000$.

\end{itemize}
In the construction of $(\bboxdot{r} \varphi)_{0011}$, we rely on $\autd_r$ being deterministic, since we quantify over all states reached by a prefix. In a non-deterministic automaton, there could be rejecting runs on accepted words, which would still be required to be completable to an accepting run by $(\bboxdot{r} \varphi)_{0011}$. 

A straightforward induction over the construction of $\varphi$, relying on the fact that $\varphi$ is limit-matching, yields the correctness of the translation. The fact that $\varphi$ is limit-matching explains the construction of $(\bboxdot{r} \varphi)_{\beta}$, which only has to implement the first case (\myquot{$\size{\Rexp(w,r)} = \infty$}) of the definition of the semantics.
\end{proof}

Now, the model checking and the synthesis problem for $\rpromptldl$, which are defined as expected, can be solved by reducing them to their analogues for $\promptldl$ (cf.~Section~\ref{subsec-promptldl}). We obtain the following results.

\begin{corollary}
The $\rpromptldl$ model checking and synthesis problem are decidable for the test-free limit-matching fragment.
\end{corollary}

We refrain from specifying the exact complexity of the algorithms, as we conjecture them to be several exponents away from optimal algorithms: The guards~$r_{q_\init,q}$ and $r_{q,F}$ are already of doubly-exponential size and we still have to translate the formula~$\varphi_\tval$ containing these guards into (deterministic) automata to solve the problems. 

Note that our approach for the fragment, which relies on a translation to $\promptldl$, cannot easily be extended to formulas with tests and to formulas with non-limit-matching guards. The existence of tests complicates the construction of the deterministic automaton required to \myquot{split} the guards. Consider, for example, the guard~$(\varphi_0?\conc a \conc \varphi_0' ) + (\varphi_1?\conc a \conc \varphi_1' )$: after processing an $a$, depending on which tests hold true before the $a$, the automaton still has to distinguish whether $\varphi_0'$ or $\varphi_1'$ has to hold after processing the $a$. Implementing this requires non-determinism that cannot be resolved while only reading a prefix of a trace. 

 Complicating the situation even further, the lack of negations in prompt logics does not allow to \myquot{disambiguate} the guard.  Similarly, allowing non-limit-matching guards requires us to implement the full case distinction in the definition of the semantics of the robust box operator. However, implementing a case distinction in $\promptldl$ is again complicated by the lack of negations.

%% file: content/conclusion.tex
We addressed the problems of verification and synthesis with robust, expressive, and quantitative linear temporal specifications.
Inspired by robust $\ltl$, we have first developed robust extensions of the logics $\ldl$ and $\prompt$, named $\rldl$ and $\rprompt$, respectively.
Then, we combined $\rldl$ and $\rprompt$ into a third logic, named $\rpromptldl$, which has the expressiveness of $\omega$-regular languages and allows robust reasoning about timing bounds.

For~$\rldl$ and~$\rprompt$, we have shown how to solve the model checking and synthesis problem relying on the exponential compilation property. Hence, all these problems are not harder than those for plain $\ltl$. 
The situation for the combination of all three basic logics, i.e., for $\rpromptldl$, is less encouraging.
We show the problems to be decidable for an important fragment, but due to a blowup of the formulas during the reduction, we (most likely) do not obtain optimal algorithms.
Decidability for the full logic remains open. 

In future work, we aim to determine the exact complexity of the model checking and synthesis problem for (full) $\rpromptldl$. One promising approach is to generalize the translation of $\rldl$ into weak alternating Büchi automata. However, this requires a suitable quantitative alternating automata model with strong closure properties that can be transformed into equivalent non-deterministic and deterministic automata. 

Another promising direction for further research is to study the semantics for the robust box operator proposed in Footnote~\ref{footnote-altsemantics} on Page~\pageref{footnote-altsemantics}. In particular, it is open whether the translation into alternating automata can be generalized to this setting without a blowup.
Also, we leave open whether full robust $\ltl$, i.e., with until and release, can be embedded into $\rldl$. As is, the robust semantics of the release operator (see~\cite{TabuadaNeider16}) is not compatible with our robust semantics for $\rldl$. In future work, we plan to study generalizations of full robust $\ltl$. 

Another natural question is whether the techniques developed for $\rldl$ can be applied to a robust version of the Property Specification Language~\cite{EisnerFismanPSL}. 

\paragraph*{Acknowledgements} We would like to the thank the reviewers for their detailed feedback, which improved the paper considerably.